\documentclass[12pt,oneside,english]{amsart}
\usepackage{mathpazo}
\usepackage[scaled=0.92]{helvet}
\usepackage{courier}
\usepackage[T1]{fontenc}
\usepackage[latin9]{inputenc}
\usepackage[a4paper, margin=1in]{geometry}
\usepackage{caption}
\usepackage{subcaption}
\usepackage{floatrow} 
\usepackage{amsfonts}
\usepackage{booktabs}
\usepackage{siunitx}
\usepackage[framemethod=tikz]{mdframed}
\usepackage{float}

\usepackage{amsthm}
\usepackage{amssymb}
\usepackage{setspace}
\usepackage{xcolor}
\usepackage[authoryear]{natbib}
\usepackage{graphicx}
\usepackage{booktabs}
\makeatletter
\numberwithin{equation}{section}
\numberwithin{figure}{section}

\theoremstyle{plain}

\newtheorem{lemma}{Lemma}[section]

\numberwithin{equation}{section}

\renewcommand{\v}[1]{\boldsymbol{#1}}

\newcommand{\xobs}{\mathbf{x}^{(\mathrm{o})}}
\newcommand{\xmis}{\mathbf{x}^{(\mathrm{m})}}
\newcommand{\xobsi}{\mathbf{x}^{(\mathrm{o})}_i}
\newcommand{\xmisi}{\mathbf{x}^{(\mathrm{m})}_i}
\newcommand{\xobshat}{\hat{\mathbf{x}}^{(\mathrm{o})}}
\newcommand{\xmishat}{\hat{\mathbf{x}}^{(\mathrm{m})}}

\newcommand{\xobstest}{\tilde{\mathbf{x}}^{(\mathrm{o})}}
\newcommand{\xmistest}{\tilde{\mathbf{x}}^{(\mathrm{m})}}

\makeatother

\usepackage{babel}
\begin{document}
\title[Bayesian Prediction with Covariates Subject to Detection Limits]{Bayesian Prediction with Covariates Subject to Detection Limits}
\author{Caroline Svahn and Mattias Villani}
\thanks{
Svahn: \textit{Department of Computer and Information Science, Linkoping University} and \textit{Ericsson AB}. 
Villani: \textit{Dept of Statistics, Stockholm University} and 
\textit{Department of Computer and Information Science, Linkoping University}.
\textit{E-mail: mattias.villani@stat.su.se}. \\
\textit{This work was partially supported by the Wallenberg AI, Autonomous Systems and
Software Program (WASP) funded by the Knut and Alice Wallenberg Foundation.}
}

\begin{abstract}
Missing values in covariates due to censoring by signal interference or lack of sensitivity in the measuring devices are common in industrial problems. We propose a full Bayesian solution to the prediction problem with an efficient Markov Chain Monte Carlo (MCMC) algorithm that updates all the censored covariate values jointly in a random scan Gibbs sampler. We show that the joint updating of missing covariate values can be at least two orders of magnitude more efficient than univariate updating. This increased efficiency is shown to be crucial for quickly learning the missing covariate values and their uncertainty in a real-time decision making context, in particular when there is substantial correlation in the posterior for the missing values. The approach is evaluated on simulated data and on data from the telecom sector. Our results show that the proposed Bayesian imputation gives substantially more accurate predictions than na\"{\i}ve imputation, and that the use of auxiliary variables in the imputation gives additional predictive power.
\end{abstract}
\maketitle


\section{Introduction}\label{sec:intro}

Covariates subject to detection limits are common in industrial settings. The values may for instance be censored due to lack of sensitivity in measuring equipment such as in biomedical data \citep{hughes99,paxton97, lyles00} or signal processing scenarios such as localization \citep{dovis2015}.

The strategies for handling missing values in regression models are mainly focused on handling missingness in the response variable, whereas much less work have been done on missing covariate values. Moreover, much of the existing literature on covariates subject to detection limits is limited to situations with a small number of covariates, and the focus is typically parameter inference rather than prediction and decision making. 

\cite{lee2018} handle data where potentially all covariates are subject to detection limits using a generalized linear model estimated with maximum likelihood. A similar approach is developed  in \cite{deLima2018} for a Gaussian linear regression model.
Multiple imputation is proposed as a fast alternative by \cite{Bernhardt2015}, \cite{lee2012} and \cite{Arunajadai2012}. The approach presented by \cite{lee2012} can handle highly censored and correlated data, however, the authors do not present a complete strategy for data where all covariates are subject to detection limits.

While frequentist approaches generally have the advantage of being relatively fast, Bayesian methods can quantify the uncertainty for both parameters and predictions in a way that is directly usable for decision making under uncertainty. This is clearly crucial in safety critical scenarios where faulty decisions have severe consequences, but also in less dramatic but often occurring decisions, such as in wireless telecommunications where a faulty decision may disconnect users from the network \citep{ryden18}. In such scenarios, providing measures of uncertainty will aid the system in the decision making and provide a more reliable network. \cite{wu2012} suggest a Bayesian generalized linear model for lower, upper or interval censored data, and \citet{yue2016bayesian} propose a Bayesian generalized mixed model. The framework in \citet{yue2016bayesian} makes interesting use of auxiliary variables to aid in the Bayesian imputation when several covariates are subject to detection limits; we will adopt the same auxiliary model in our work here. 

The existing Bayesian literature use Gibbs sampling algorithms to simulate from the joint posterior of the model parameters and the missing covariate values. The proposed Gibbs samplers update the missing covariate values in an observation conditional on all other missing values, see e.g. \citet{yue2016bayesian} and \citet{wu2012}. We show that this can be very inefficient when the missing values are highly correlated in the posterior. Efficient sampling of missing values is particularly important in the prediction phase where the missing covariates for a new observation must be learned quickly in a real-time context. We therefore develop a fast and efficient Markov Chain Monte Carlo (MCMC) algorithm that samples all missing covariates jointly. The joint sampling is performed using the recently proposed and highly efficient truncated multivariate normal simulation algorithm in \citet{botev2017normal} and we additionally propose a random scan implementation \citep{amit1991comparing} to further increase the speed of the missing covariate updating step. The algorithm is presented for the regression case, but can equally well be used for classification via the data augmentation device for the Probit model in \citet{albert1993bayesian} or in logistic regression via the P{\'o}lya-Gamma augmentation in \citet{polson2013bayesian}. 

In contrast to \cite{wu2012} and \citet{yue2016bayesian}, our focus is on probabilistic prediction and decision making. We highlight how missing covariate values in the test data makes it necessary to re-run the MCMC updating steps also for the missing values in the training data. This is computationally demanding and we therefore propose and assess a batch mode strategy to circumvent this to speed up predictions in a real-time setting.

The paper is structured as follows. Section \ref{sec:model} presents the regression model where covariates are missing due to censoring at a detection limit and the prior distribution for the model parameters. Section \ref{sec:inference} presents the computational Bayesian inference algorithm and importance of joint updating of missing values. Section \ref{sec:prediction} develops the Bayesian predictive framework and discusses computational considerations for real-time applications. In Section \ref{sec:simulations} we present results for artificial data and in Section \ref{sec:application} we evaluate the model on data from a simulator of telecom signals in a wireless network. Section \ref{sec:discussion} concludes and outlines some directions for future research.

\section{Regression with covariates subject to detection limits}\label{sec:model}

We consider the linear regression
\begin{equation}\label{eq:model}
    y_i=\beta_0+\tilde{\v\beta}^{\top} \v x_i+\varepsilon_i,\hspace{1em}\varepsilon_i\overset{iid}{\sim} \mathcal{N}(0,\sigma^2)\hspace{1em}\text{for } i=1,\ldots,n,
\end{equation}
where $y_{i}$ is the response variable and $\v x_i = (x_{i1},\ldots,x_{ip})^\top$ denotes the vector of covariates for the $i$th observation. 

The covariate value $x_{ij}$ is not observed if it falls below a detection limit $c_{ij}$. Depending on the application, the detection limit may be: i) different for different covariates but the same across observations, $c_{ij}=c_j$ for all $i$, or ii) the same for each covariate but vary across observations, $c_{ij}=c_i$ for all $i$. The first setting is exemplified in \citet{yue2016bayesian} where the covariates are levels of vitamin D2 and D3 that are each subject to known detection limits. The latter setting is common for signal data, e.g. telecommunication signals, where the covariates are different signal strength measurements  and the weaker signals may be drowned by the strongest signal. The detection limit is here determined by the strongest signal and is therefore the same for all covariates in a given observation, but will differ across observations. We will take all $c_{ij}$ to be known here (once the observed covariates have been collected) and discuss the case with unknown detection limits in Section \ref{sec:discussion}. 

Following \citet{yue2016bayesian} we model the covariate vector for the complete data as a multivariate regression model given a vector of $r$ completely observed auxiliary covariates $\v w_{i}$:
\begin{equation}
\v x_i=\v\Gamma^{\top} \v w_i+\v v_i\hspace{1em}\v v_i\overset{iid}{\sim} \mathcal{N} (0,\v\Omega),
\end{equation}
where $\v\Gamma$ is the $r\times p$ matrix of regression coefficients (including intercepts) and
$\v\Omega$ is the $p\times p$ positive definite error covariance matrix. The observations falling below the detection limit in $\v x_i$ will be inferred using the information in both $\v w_{i}$ (when available) and the observed covariates in $\v x_i$.

The complete model can be written
\begin{align}\label{eq:multi_model}
     y_i &= \beta_0 + \tilde{\v\beta}^{\top} \v x_i+\varepsilon_i,\hspace{1em}\varepsilon_i\overset{iid}{\sim} \mathcal{N}(0,\sigma^2). \\
    \v x_i &= \v\Gamma^{\top} \v w_i+\v v_i,\hspace{1em}\v v_i\overset{iid}{\sim} \mathcal{N} (\v 0,\v\Omega),
\end{align}
where $x_{ij}$ is unobserved if $x_{ij}<c_{ij}$. We will use the notation where the complete $p$-dimensional covariate vector for the $i$th observation, $\mathbf{x}_i$, is decomposed into $p_\mathrm{o}$ observed covariate values $\mathbf{x}^{(\mathrm{o})}_i$ and $p_\mathrm{m}$ missing covariate values $\mathbf{x}^{(\mathrm{m})}_i$, with the understanding that either of these subvectors can be empty.

We take a Bayesian approach and assume the following prior with independence between the blocks of parameters, except $\v \Gamma$ and $\v \Omega$,
\begin{align*}
\beta_0 & \sim \mathcal{N}(0,\tau_{\beta_0}^{2}) \\
\tilde{\v \beta} & \sim \mathcal{N}(0,\tau_{\tilde \beta}^{2}\v I_{p})\\
\sigma^{2} & \sim\mathrm{IG}(a,b)\\
\v \gamma \vert \v \Omega & \sim \mathcal{N}(\v 0,\v \Omega \otimes \tau_{\gamma}^{2}\v I_{r})\\
\v \Omega & \sim\mathrm{IW}(\v A, \kappa),
\end{align*}
where $\gamma = \operatorname{vec}\v \Gamma$ stacks the columns of $\v \Gamma$ in a vector, $\mathrm{IG}(a,b)$ is the inverse Gamma distribution and $\mathrm{IW}(\v A, \kappa)$ is the inverse
Wishart distribution with $\kappa$ degrees of freedom. The prior settings used in the experiments and in the telecom application can be found in Appendix \ref{Appendix:data}.

\section{Bayesian inference}\label{sec:inference}

This section presents the Gibbs sampling algorithm for sampling from the joint posterior distribution of all the model parameters and the missing covariate values. Our proposed method samples all missing values jointly and we show the importance of this feature, particularly for real-time/online prediction. 

\subsection{Gibbs sampling}\label{subsec:Gibbs}

The joint posterior distribution of the model parameters and missing covariates is intractable. We sample from the joint posterior using Gibbs sampling with the following updating steps.

\subsubsection*{\textbf{Updating} $\boldsymbol{\beta}$} The full conditional posterior of $\boldsymbol{\beta}$ is
\begin{equation*}
    \v\beta | \v y, \v X, \v W, \sigma^2, \v  \gamma, \v \Omega \sim \mathcal{N}(\v \mu_{\beta}, \v \Sigma_{\beta}),
\end{equation*}
where $\v \beta = \big( \beta_0,\tilde{\v \beta}^\top \big) ^\top$,
\begin{align*}
    \v \mu_{\beta} & = \Bigg(\frac{1}{\sigma^2}\v X^\top \v X+\v D^{-1}\Bigg)^{-1} \Bigg(\frac{1}{\sigma^2} \v X^\top \v y\Bigg) \\
    \Sigma_{\beta} & = \Bigg(\frac{1}{\sigma^2}\v X^\top \v X+\v D^{-1}\Bigg)^{-1}
\end{align*}
$\v y = (y_1,\ldots,y_n)^\top$, $\v X = (\v x_1,\ldots,\v x_n)^\top$, $\v W = (\v w_1,\ldots,\v w_n)^\top$,
and
\begin{align*}
\v D & = \begin{bmatrix} \tau_{\beta_0}^2 & \v0\\ 
                        \v0 & \tau_{\tilde{\beta}}^2 \cdot \v I_p
        \end{bmatrix}.
\end{align*}

\subsubsection*{\textbf{Updating} $\boldsymbol{\sigma^2}$} The full conditional posterior for $\sigma^2$ is

\begin{equation*}
   \sigma^2 |  \v y, \v X, \v W, \v \beta, \v  \gamma, \v \Omega \sim \text{IG}(\tilde a, \tilde b)
\end{equation*}
where
\begin{align*}
    \tilde a & = \frac{n}{2} + a \\
    \tilde b & = \frac{(\v y - \v X \v \beta)^\top(\v y - \v X \v \beta)}{2}+b.
\end{align*}

\subsubsection*{\textbf{Updating} $\v \Omega$ \textbf{and} $\boldsymbol{\Gamma}$} 
The joint conditional posterior of $\v \Omega$ and $\gamma = \operatorname{vec}\v \Gamma$ is

\begin{align*}
    \v \Omega| \v y, \v X, \v W, \v \beta, \sigma^2 \sim & \text{ IW}\Big(\tilde{\v A}, n + \kappa \Big) \\
   \v \gamma  | \v \Omega, \v y, \v X, \v W, \v \beta, \sigma^2 \sim & \mathcal{ N}_{pr}\Bigg[
                                                    \tilde{\v \gamma}, \v \Omega \otimes \big(\tau^{-2}_{\v \gamma}\v I_r + \v W^{\top} \v W \big)^{-1} \Bigg],
\end{align*} 
where
\begin{align*}
    \tilde{\v A} &=   \v A + \tau^{-2}_{\v \gamma} \tilde{\v \Gamma}^{\top} 
                                \tilde{\v \Gamma} + \big(\v X - \v W\tilde{\v \Gamma}\big)^{\top}\big(\v X - \v W\tilde{\v \Gamma}\big)\\
    \tilde{\v \Gamma} &=  \big(\tau^{-2}_{\v \gamma}\v I_r + \v W^{\top}\v W\big)^{-1} \v W^{\top}\v X,
\end{align*}
and $\tilde{\v \gamma} = \operatorname{vec}(\tilde{\v \Gamma})$.

When auxiliary variables are not available, the auxiliary regression has only an intercept and we obtain the following special case of the above result by inserting $\v W = (1,\ldots,1)^\top$
\begin{align*}
    \v \Omega| \v y, \v X, \v \beta, \sigma^2 \sim & \text{IW}\Big(\tilde{\v A}, n + \kappa \Big) \\
   \v \gamma  | \v \Omega, \v y, \v X, \v \beta, \sigma^2 \sim & \mathcal{ N}_{pr}\Bigg[
                                                    \tilde{\v \gamma}, \big(\tau^{-2}_{\v \gamma} + n\big)^{-1}\v \Omega \Bigg],
\end{align*} 
where
\begin{align*}
    \tilde{\v A} &=  \v A + \tau^{-2}_{\v \gamma} \tilde{\v \gamma}^{\top} 
                                \tilde{\v \gamma} + \sum_{i=1}^n  ( \v x_i- \tilde{\v \gamma})( \v x_i-\tilde{\v \gamma})^\top\\
    \tilde{\v \gamma} &=  \big(\tau^{-2}_{\v \gamma} + n\big)^{-1} \sum_{i=1}^n \v x_i.
\end{align*}

\subsubsection*{\textbf{Updating} $\xmis_1,\ldots,\xmis_n$}

The missing values in a given observation, $\xmisi$, are conditionally independent of the missing values in all other observations. Each $\xmisi$ vector can therefore be drawn in parallel from truncated multivariate normal distributions:

\begin{equation*}
    \xmisi |\xobsi, \v y, \v W, \v \beta, \sigma^2, \v \gamma, \v \Omega \sim \mathcal{N}(\boldsymbol{\mu}_{\xmisi}, \boldsymbol{\Omega}_{\xmisi}; \xmisi \leq \v c_i),
\end{equation*}
where $\v c_i$ is a vector of detection limits for $\xmisi$,
\begin{align*}
     \boldsymbol{\mu}_{\xmisi} & = \Bigg(\frac{1}{\sigma^2}\v \beta_\mathrm{m} \v \beta_\mathrm{m}^\top + \bar{\v \Sigma}_i^{-1}\Bigg)^{-1} \Bigg(\bar{\v \Sigma}_i^{-1} \bar{\v \mu}_i+\frac{\tilde{y}_i}{\sigma^2} \v \beta_\mathrm{m}\Bigg)\\
    \boldsymbol{\Omega}_{\xmisi} & = \Bigg(\frac{1}{\sigma^2}\v \beta_\mathrm{m} \v \beta_\mathrm{m}^\top + \bar{\v \Sigma}_i^{-1}\Bigg)^{-1},
\end{align*}
$\tilde{y}_i = y_i - \beta_0 - \v \beta_\mathrm{o}^\top \xobsi$ and $\boldsymbol{\beta}_{\mathrm{m}}$ and $\boldsymbol{\beta}_{\mathrm{o}}$ are the subsets of $\v \beta$ corresponding to the missing and observed values in $\mathbf{x}_i$, respectively. Finally, 
\begin{align*}
\bar{\v \mu}_i &= \xmishat_i +\v \Omega_{\mathrm{m}\mathrm{o}}\v \Omega_{\mathrm{o}\mathrm{o}}^{-1}(\xobsi-\xobshat_i) \\
\bar{\v \Sigma}_i &= \v \Omega_{\mathrm{m}\mathrm{m}} - \v \Omega_{\mathrm{m}\mathrm{o}}\v \Omega_{\mathrm{o}\mathrm{o}}^{-1}\v \Omega_{\mathrm{o}\mathrm{m}},
\end{align*}
with $\hat{\v x}_i = (\xmishat_i{}^\top,\xobshat_i{}^\top)^\top = \v \Gamma^\top \v w_i$ and the corresponding decomposition of
\begin{equation*}
    \v \Omega =
    \begin{pmatrix}
        \v \Omega_{\mathrm{m}\mathrm{m}} & \v \Omega_{\mathrm{m}\mathrm{o}}\\
        \v \Omega_{\mathrm{o}\mathrm{m}} & \v \Omega_{\mathrm{o}\mathrm{o}}
    \end{pmatrix},
\end{equation*}
subject to the reordering of rows and columns so that the missing values are first in $\v x_i$.

We use the efficient algorithm for sampling from truncated multivariate normal distributions recently developed in \citet{botev2017normal}.

\subsection{Joint sampling of missing values}

Our Gibbs sampling algorithm in Section \ref{subsec:Gibbs} draws all missing values in an observation jointly rather than drawing each missing value conditional on the other missing values in the observation, as in earlier literature, see e.g. \cite{yue2016bayesian}. The gains from a multivariate approach depends on how correlated the missing values are in the posterior distribution, which varies from observation to observation. This is important since we need to sample all missing covariate values in the test observation in the prediction step, and computing time is often crucial at prediction time. As shown below, a Gibbs sampler with univariate updating steps will require a substantially larger number of iterations to obtain the same precision for a test observation with large posterior correlations between missing values. 

Lemma \ref{lem:jointsampling} derives the correlation matrix in the full conditional posterior of the missing values $\xmis$ for an arbitrary observation, to explore when substantial posterior correlation occurs. The lemma shows that the posterior correlation can be substantial in a number of different scenarios depending on how correlated the $\xmis$ are conditional on $\xobs$, and on the correlation between each missing value and $y$, conditional on $\xobs$.

\begin{lemma}\label{lem:jointsampling}
Let $\rho(\xmis|\xobs, y)$ denote the correlation
matrix in the full conditional posterior for the missing covariates
$\xmis=( x^{(\mathrm{m})}_1,...,x^{(\mathrm{m})}_{p_\mathrm{m}})$ in a given observation,
conditional on the observed covariates $\xobs$ and all
model parameters $\v \beta,\sigma^2,\v \Gamma$ and $\v \Omega$; The conditioning on the model parameters are suppressed in the notation
for clarity. Then, 
\[
\rho(\xmis|\xobs, y)=\boldsymbol{R}\left(\rho(\xmis|\xobs)-\rho(\xmis,y|\xobs)\rho(\xmis,y|\xobs){}^{\top}\right)\boldsymbol{R}
\]
where $\rho(\xmis\vert\xobs)$ is the correlation matrix among $\xmis$ conditional on $\xobs$ but not on $y$, $\rho(y,\xmis|\xobs)$ is the
column vector of conditional correlations between $y$ and each missing covariate and
$$\boldsymbol{R}=\mathrm{Diag}\left(1/\sqrt{\left(1-\rho^{2}(x^{(\mathrm{m})}_1,y|\xobs)\right)},\ldots,1/\sqrt{\left(1-\rho^{2}(x^{(\mathrm{m})}_{p_\mathrm{m}},y|\xobs)\right)}\right).
$$
\end{lemma}

\begin{proof}
See Appendix \ref{Appendix:Proofs}.
\end{proof}

Table \ref{tab:botev} shows results from a simulation experiment comparing univariate versus joint sampling of the missing covariates. The table presents the effective sample size (ESS) ratio $\mathrm{ESS}_\mathrm{multi}/\mathrm {ESS}_\mathrm{uni}$ for the missing covariate values in five simulated data sets with $n=1000$ observations on $p = 40$ covariates and approximately 40 \% censored values. See Appendix \ref{Appendix:data} for more details about the data generating process. It can be seen that joint sampling is never less than half as efficient as univariate sampling, while it can be at least two orders of magnitude as efficient for some missing values compared to univariate updating. Furthermore, a joint update of all missing values is on average twice as fast as a set of univariate updating steps in this simulation. \\

\begin{table}[]
    \centering
 \begin{tabular}[!htb]{ccrrrrr}\toprule 
  & & \multicolumn{5}{c}{Quantiles}   \\
    Dataset & & 0\%  &  25\%      &     50\%   &    75\%   & 100\% \\
    \midrule
    1 & & 0.52 &  1.03 & 1.25 & 6.61 & 43.8 \\
    2 & & 0.51 & 1.04 & 1.29 & 9.08 & 65.6 \\
    3 & & 0.54 & 1.03 & 1.22 & 6.62 & 61.1 \\
    4 & & 0.48 & 1.03 & 1.24 & 8.00 & 57.7 \\
    5 & & 0.55 &  1.03 &  1.25 & 10.13 & 141.7 \\ \bottomrule
\end{tabular}    \caption{Ratios of the effective sample size (ESS) comparing joint sampling of missing values to univariate updates. Each row represents the quantiles of the ratio $\mathrm{ESS}_\mathrm{multi}/\mathrm {ESS}_\mathrm{uni}$ for a simulated dataset with $n=1000$ observations on $p = 40$ covariates and approximately 40 \% censored values. The computing time for a joint update is approximately twice as fast as a sequence of univariate updates.}
    \label{tab:botev}
\end{table}

\subsection{Increasing efficiency by using Random scan}


In this subsection we explore the use of a random scan update \citep{amit1991comparing} of the missing values. A random scan sampler updates each $\xmisi$ vector with a fixed probability at each Gibbs iteration to increase the sampling efficiency for a given time budget. A small update probability will make each Gibbs iteration fast, but also less efficient since many $\xmisi$ are left unchanged at the iteration. 

Figure \ref{fig:ess} uses the simulation setup in Section \ref{sec:simulations} to illustrate the ratio of ESS per time step for the random scan Gibbs sampler compared to ESS per time step for the regular Gibbs sampler, as a function of the update probability. The top left plot is for the average ESS ratio for the predictive distributions, the top right plot for the $\v \beta$ coefficients and the average ESS for the missing values can be seen in the bottom plot. The ESS ratio for the missing values starts to deteriorate for low update probabilities, whereas the ESS ratio for the predictive distributions indicates that it is more efficient to use random scan with a rather low update probability. Of course, a too small update probability risks the convergence of the Gibbs sampler in realistic computing times, and we will use an update probability of $0.2$ as a compromise to achieve high efficiency for $\v \beta$ and the predictive distributions without sacrificing efficiency for the missing values.

\begin{figure}[!h]
\centering
\begin{subfigure}{.45\linewidth}
    \centering
    \includegraphics[width=\textwidth]{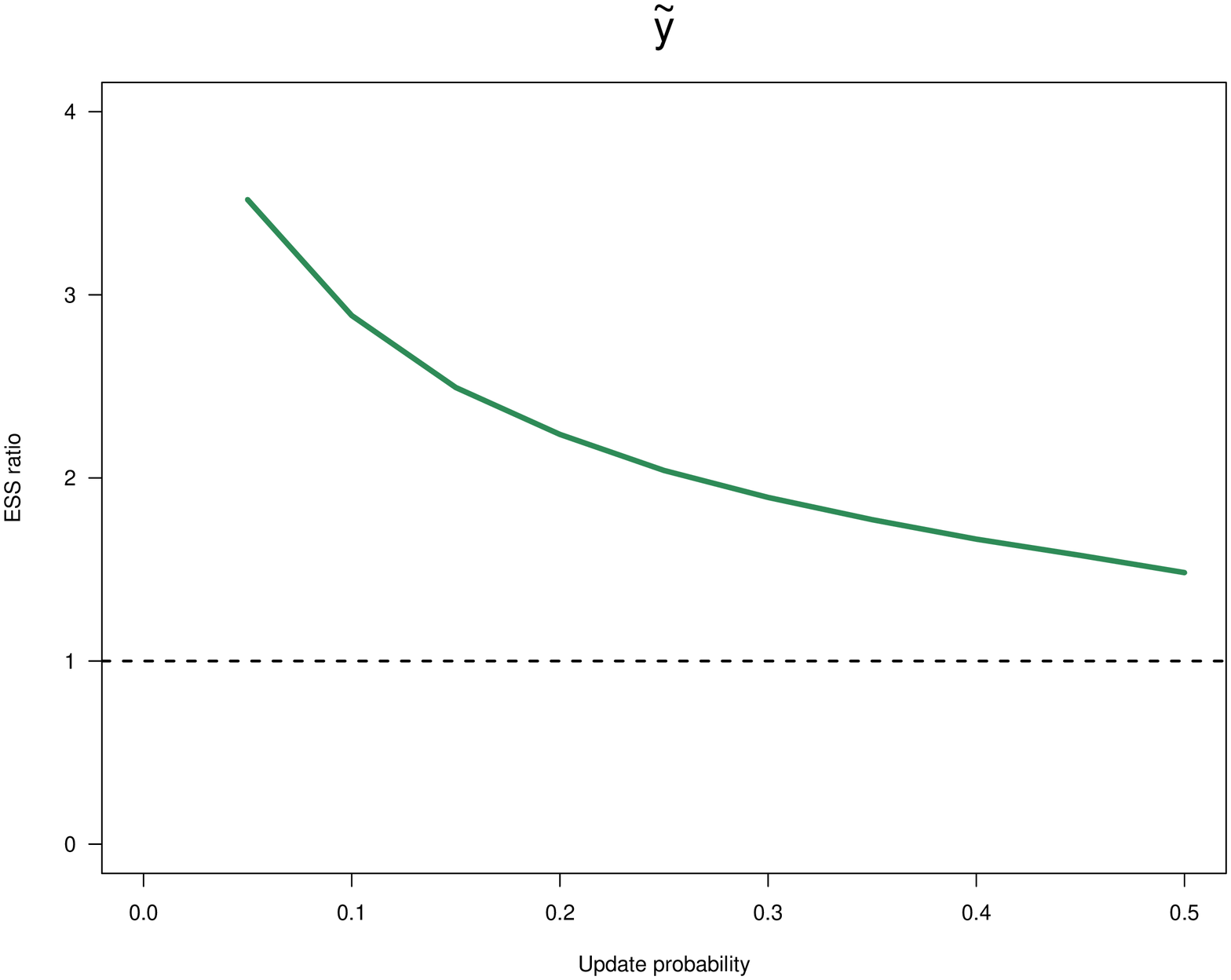}
\end{subfigure}
\begin{subfigure}{.45\linewidth}
    \centering
    \includegraphics[width=\textwidth]{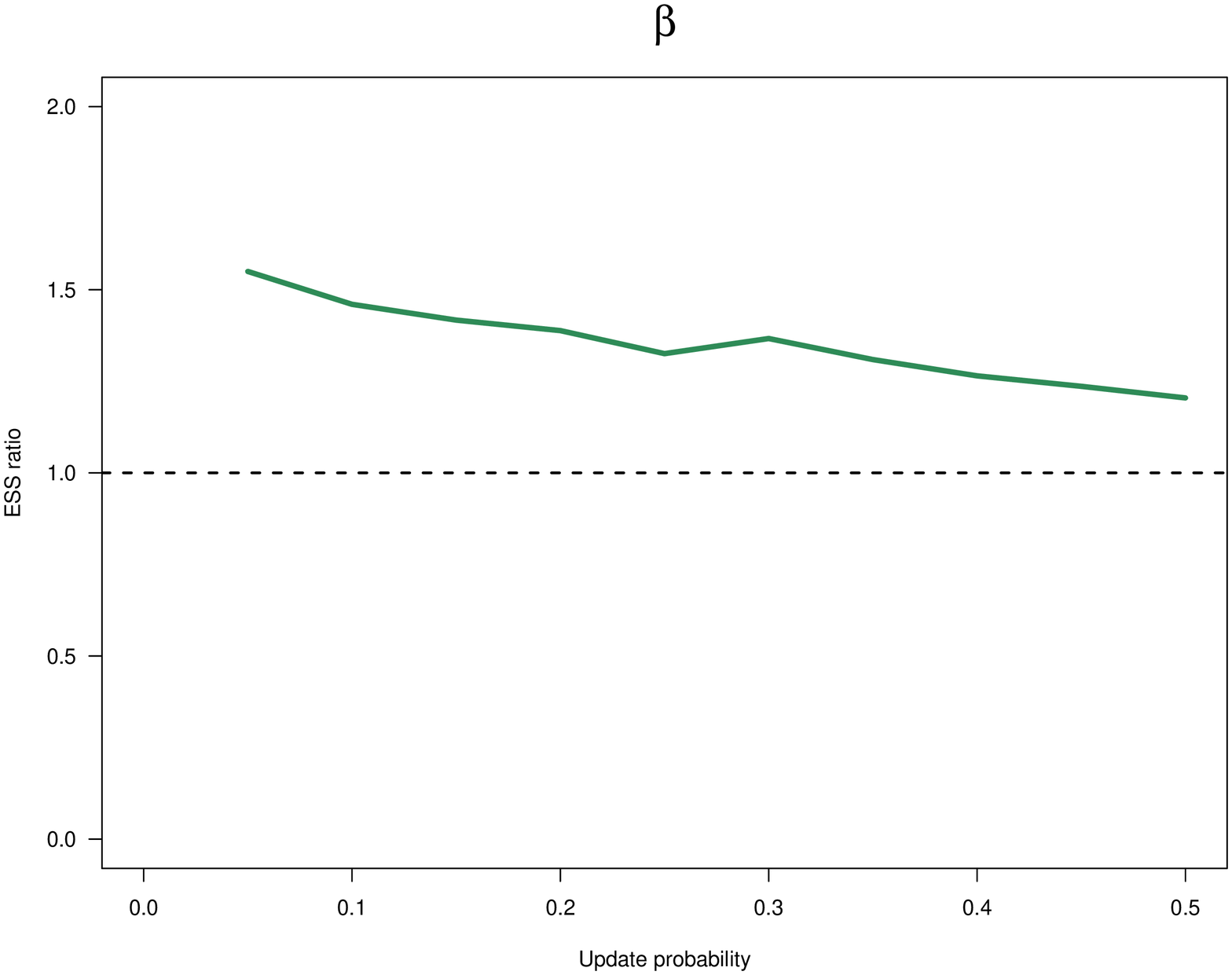}
\end{subfigure}

\begin{subfigure}{.45\linewidth}
    \centering
    \includegraphics[width=\textwidth]{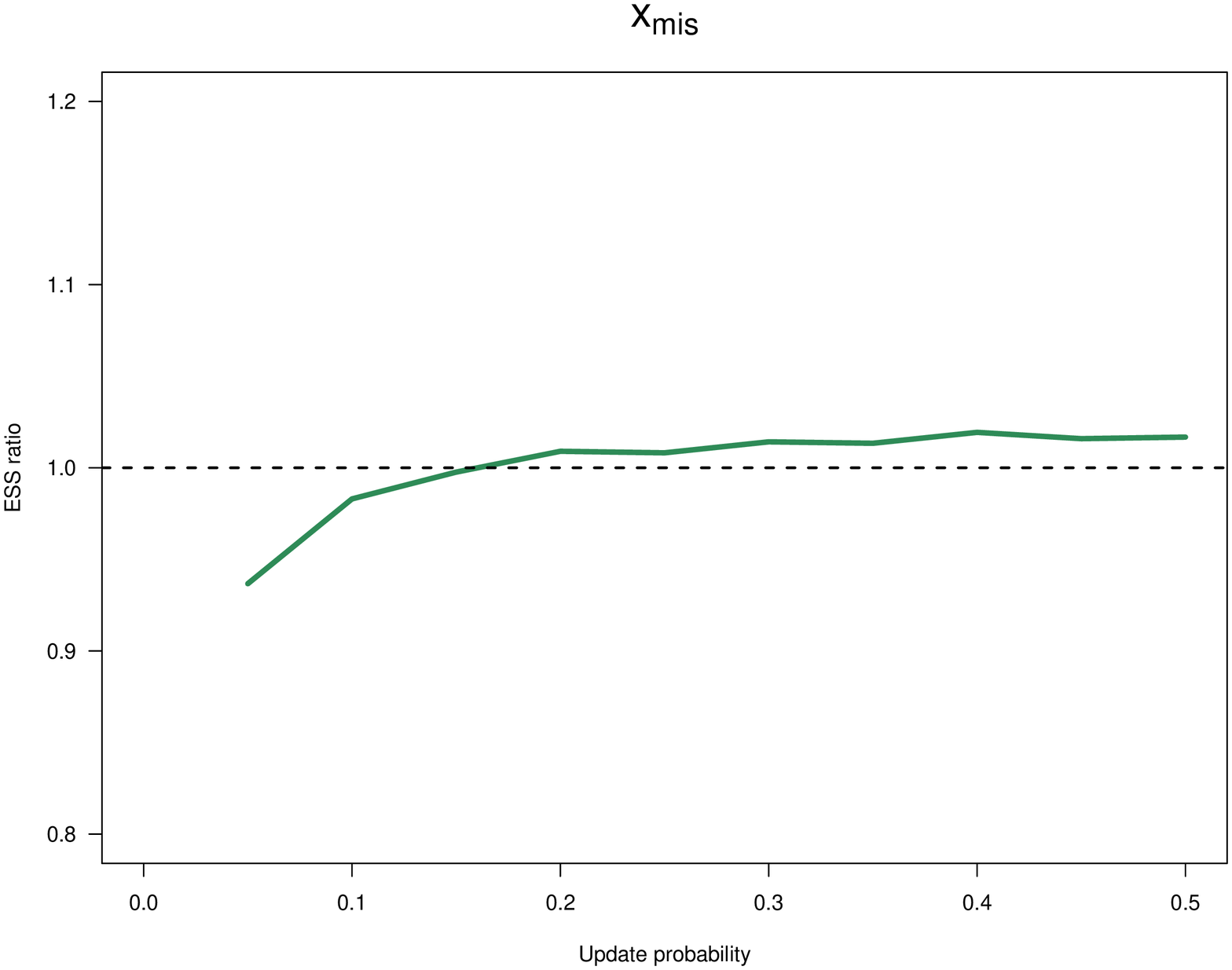}
\end{subfigure}

\caption{Ratio of average ESS per time step for the random scan Gibbs sampler compared to ESS per time step for the regular Gibbs sampler, as a function of the update probability, for the predictive distributions (top left), the $\v \beta$ coefficient distributions (top right) and the missing value distributions (bottom).}\label{fig:ess}
\end{figure}
%

\section{Online predictive distributions}\label{sec:prediction}

This section derives the predictive distribution and discusses a computational complication that surfaces when covariate values are missing in the test data. 

\subsection{Predictive distribution}
Consider first the case with complete data. Let $\tilde{\v y}$ denote the vector of $\tilde n$ response observations in the test data and $\tilde{\v X}$ the corresponding covariate values. The posterior predictive distribution is then
\begin{align*}
    p(\tilde{\v y}|\v y, \tilde{\v X},  \v X) = & \int p(\tilde{\v y}| \tilde{\v X}, \v \beta, \sigma^2)p(\v \beta, \sigma^2 | \v y, \v X) \text{ d} \v \beta \text{ d} \sigma^2.
\end{align*}
We can easily obtain samples from this predictive distribution by 
drawing parameters from the posterior $p(\v \beta, \sigma^2 | \v y, \v X)$ and for each parameter draw simulate from the model $p(\tilde{\v y}| \tilde{\v X}, \v \beta, \sigma^2)$. The posterior $p(\v \beta, \sigma^2 | \v y, \v X)$ remains fixed regardless of how much test data are available.

Missing covariate values introduce a complication: to obtain the exact posterior predictive distribution for a new batch of test data we need to run the Gibbs sampler for all the data, both training and test data. The reason for having to revisit the training data in the test stage is that the observed covariate observations in the test data gives information about the $\v \Gamma$ and $\v \Omega$, thereby bringing about changes in the posterior for the missing values in the training data which in turn affects the posterior distribution for $\v \beta$ and $\sigma$, and therefore finally the predictions for $\tilde{\v y}$.

To see this formally, let $\xmistest$ and $\xobstest$ denote the missing and observed covariates in a new test observation $(\tilde y, \tilde{\v x})$. The joint predictive distribution of $\tilde{y}$ and $\xmistest$ is then
\begin{align*}
    p(\tilde{y}, \xmistest|\v y, \v X^{(\mathrm{o})}, \xobstest) &=  \int p(\tilde{y}, \xmistest| \xobstest, \v \beta, \sigma^2, \v \Gamma, \v \Omega) \\
    & \hspace{0.5cm} \times p(\xmis,\v \beta, \sigma^2,\v \Gamma, \v \Omega | \v y, \v X^{(\mathrm{o})}, \xobstest) \text{ d} \xmis \text{ d} \v \beta \text{ d} \sigma^2 \text{ d} \v \Gamma \text{ d} \v \Omega. \\
    & = \int p(\tilde{y} | \xmistest, \xobstest,\v  \beta, \sigma^2) p(\xmistest | \xobstest, \v \Gamma, \v \Omega) \\
    & \hspace{0.5cm} \times p(\xmis,\v \beta, \sigma^2,\v \Gamma, \v \Omega | \v y, \v X^{(\mathrm{o})}, \xobstest) \text{ d} \xmis \text{ d} \v \beta \text{ d} \sigma^2 \text{ d} \v \Gamma \text{ d} \v \Omega.
\end{align*}
Hence, simulation from the predictive distribution of $\tilde{y}$ can be performed by:
\begin{enumerate}
    \item Gibbs sampling parameters and missing values in the training data from $$p(\xmis,\v \beta, \sigma^2,\v \Gamma, \v \Omega | \v y, \v X^{(\mathrm{o})}, \xobstest)$$
    \item Simulating missing values for the test observation from 
    $$\xmistest \vert \xobstest, \v \Gamma, \v \Omega \sim \mathcal{N}( \bar{\v \mu}, \bar{\v \Omega}; \xmistest \leq \v c)$$ 
    \item Simulating the prediction from $$\tilde{y} \vert \xmistest, \xobstest,\v  \beta, \sigma^2 \sim \mathcal{N}(\v \beta_\mathrm{m}^\top \xmistest + \v \beta_\mathrm{o}^\top \xobstest, \sigma^2)$$
\end{enumerate}

The Gibbs updates in Step 1 for $\v \Gamma$ and $\v \Omega$ have $n + 1$ data points, whereas the updating steps for $\v \beta$ and $\sigma^2$ are effectively only based on the $n$ data points in the training data, since $\tilde{y}$ is not observed in test. Note also that the simulation of $\xmistest$ in Step 2 is not conditional on $\tilde{y}$ and is therefore different (simpler) than Gibbs sampling update for $\xmis$ in Step 1. All of the above apply also to the case with $\tilde n >1$ test observations, where the Gibbs updates in Step 1 for $\v \Gamma$ and $\v \Omega$ have $n + \tilde n$ data points and so on.

\subsection{A computational strategy for fast online predictions}
Note that in Step 1 in the previous subsection we need to re-run the Gibbs sampler also over the training data anytime we get a new test observation for prediction; this will clearly be computationally demanding for large training data. It is therefore of interest to explore the consequences of using the approximate but computationally cheaper strategy of using the posterior of the model parameters from only the training data when inferring the missing covariates and making prediction for the test cases. This suggests a practical batch mode strategy where the Gibbs sampler is infrequently re-run on all training and test data, e.g. overnight in high frequency streams, and then using the cheaper approximate strategy in between such re-estimation checkpoints.

We run the following experiment to illustrate how often we need to re-run the Gibbs sampler on all training data. At a given time, let there be $n$ training observations available when a batch of $\tilde n + 1$ test observations arrives, and assume for simplicity that the interest is in the prediction for the last observation in this batch; we call this the \emph{observation of interest}. We are interested in comparing the following two predictive distributions:
\begin{itemize}
    \item \textbf{Exact strategy}: Re-run the Gibbs sampler for all data, training and test, and sample from the predictive distribution for the observation of interest. 
    \item \textbf{Approximate strategy}: Use the posterior distribution of the model parameters from the training data only. Sample the missing covariates for the observation of interest and compute the predictive distribution of its response.
\end{itemize}
We let the observation of interest remain the same throughout to be able to compare across different combinations of $n$ and $\tilde n$. We use $n = 1000$ and $\tilde n = d\cdot n \text{ for } d \in \{0.1, 0.5, 1, 3, 5\}$ to successively increase the test size in the experiment. We split the $p = 5$ covariates into groups of three and two covariates, respectively, with correlation $\rho = 0.8$ within the groups and $\rho = 0$ between the groups. For more details on the data generation, see Appendix \ref{Appendix:data}. Finally, we let two covariates be insignificant in the regression model in \eqref{eq:model}. The point of interest has three missing values, where the last two have non-zero $\v \beta$ coefficients in the data generating model.

Figure \ref{fig:dirtystrategymissingparam} displays the posterior densities for the means ($\v \gamma$) and the variances (diagonal elements of $\v \Omega$) for the missing values in the observation of interest. The black density corresponds to the approximate strategy where the Gibbs sampler is not revisiting the training data as new test cases come in. The densities on the color scale from red to yellow show how the posterior changes as we get increasingly more test data before the observation of interest. Figure \ref{fig:test_dens_pred} shows the effect on the predictive density. Despite some differences in $\v \gamma$ and $\v \Omega$ in Figure \ref{fig:dirtystrategymissingparam}, all predictive densities in \ref{fig:test_dens_pred} are similar. This is to be expected since the missing values are merely one of several ingredients in the predictive distribution, and also not all covariates are significant. The results in Figure \ref{fig:test_dens_pred} is only an illustration for a particular test case, but we have found in other simulations that it is quite sufficient to revisit the training data in the Gibbs sampler rather infrequently when the focus is prediction, at least when the training data is not smaller than the test data. For the remainder of this paper, the approximate strategy will therefore be used as it is accurate enough in our setting.

\begin{figure}[!h]
\centering
\begin{subfigure}{.45\linewidth}
    \centering
    \includegraphics[width=\textwidth]{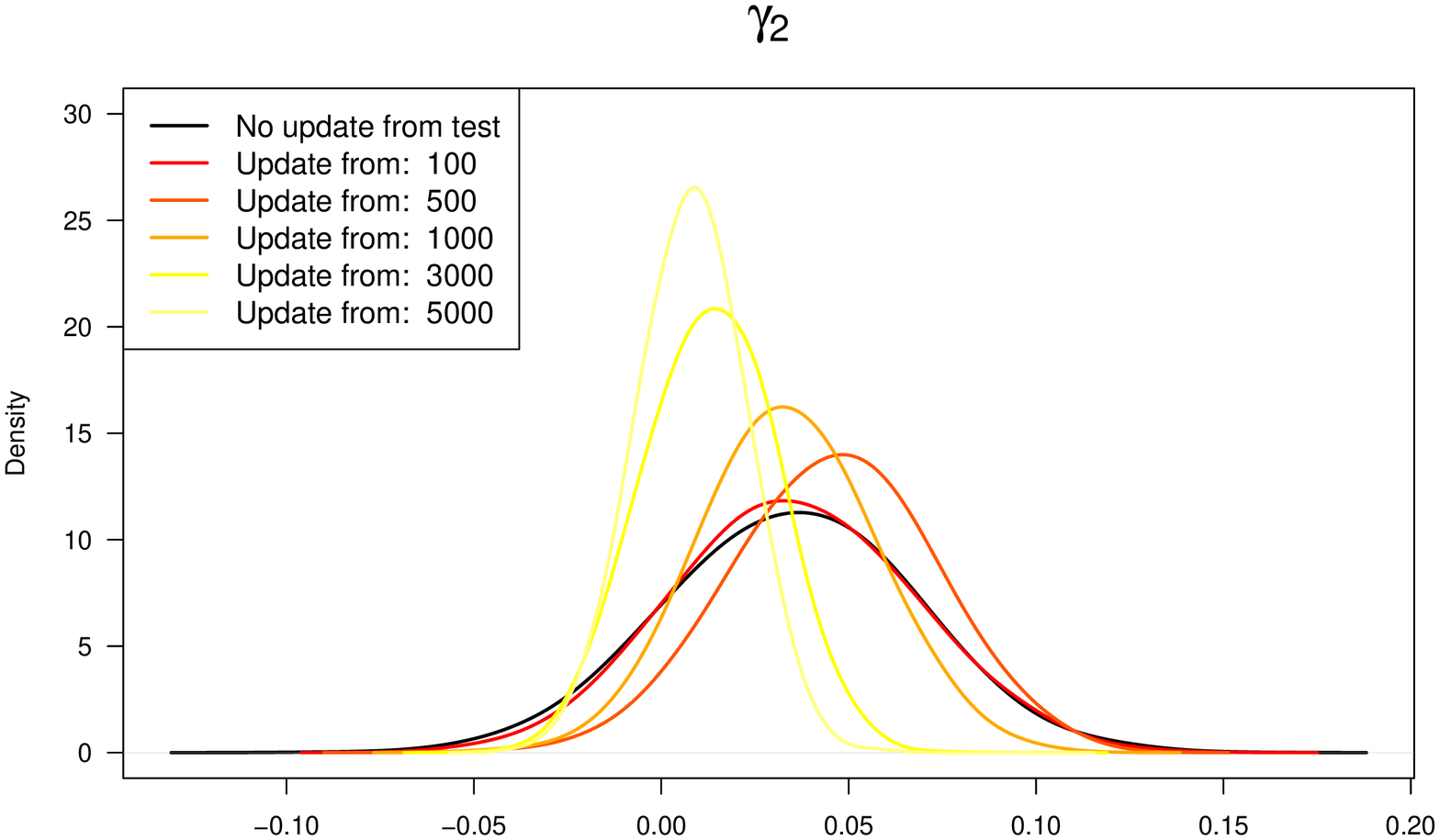}
\end{subfigure}
   \hfill
\begin{subfigure}{.45\linewidth}
    \centering
    \includegraphics[width=\textwidth]{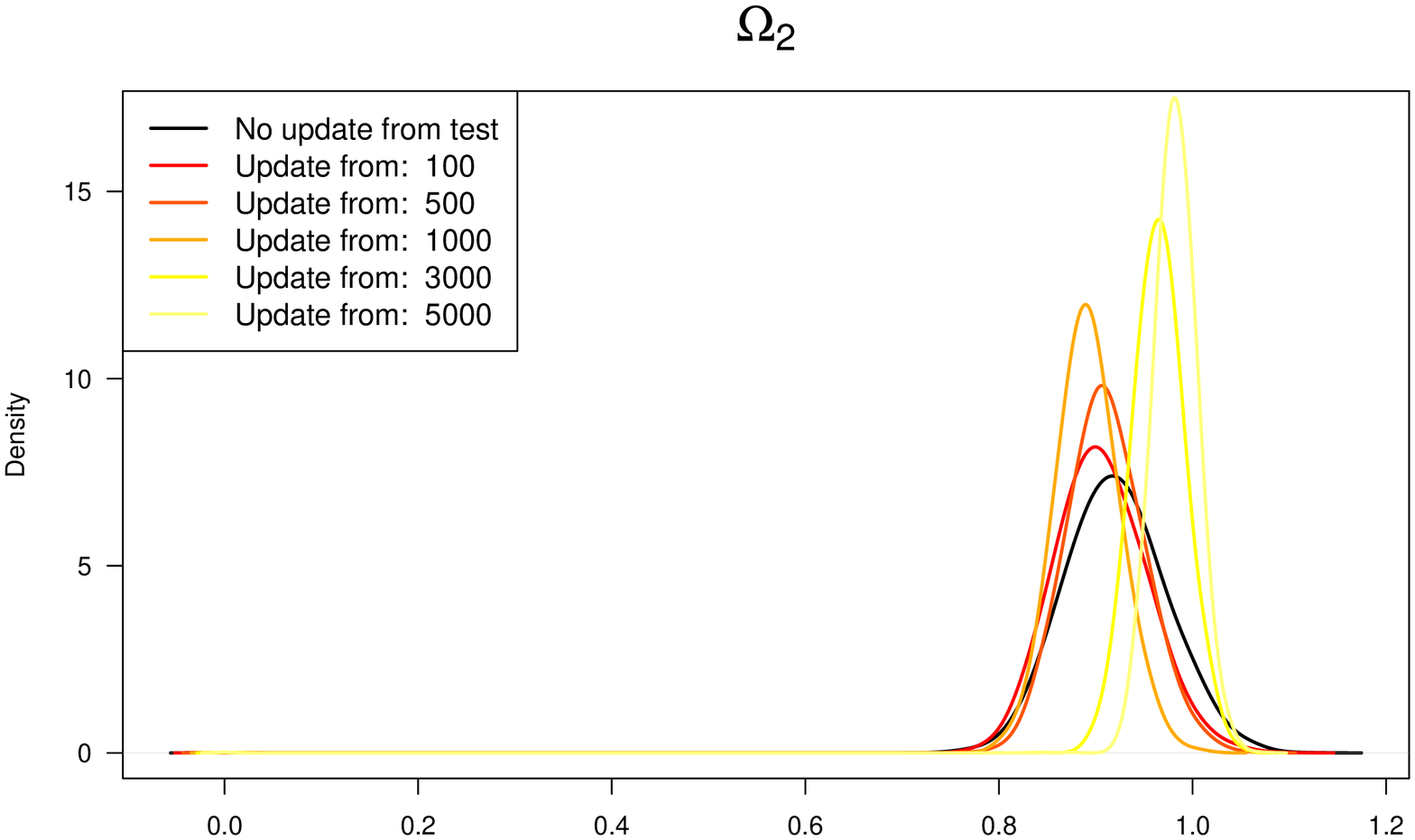}
\end{subfigure}

\begin{subfigure}{0.45\linewidth}
  \centering
    \includegraphics[width=\textwidth]{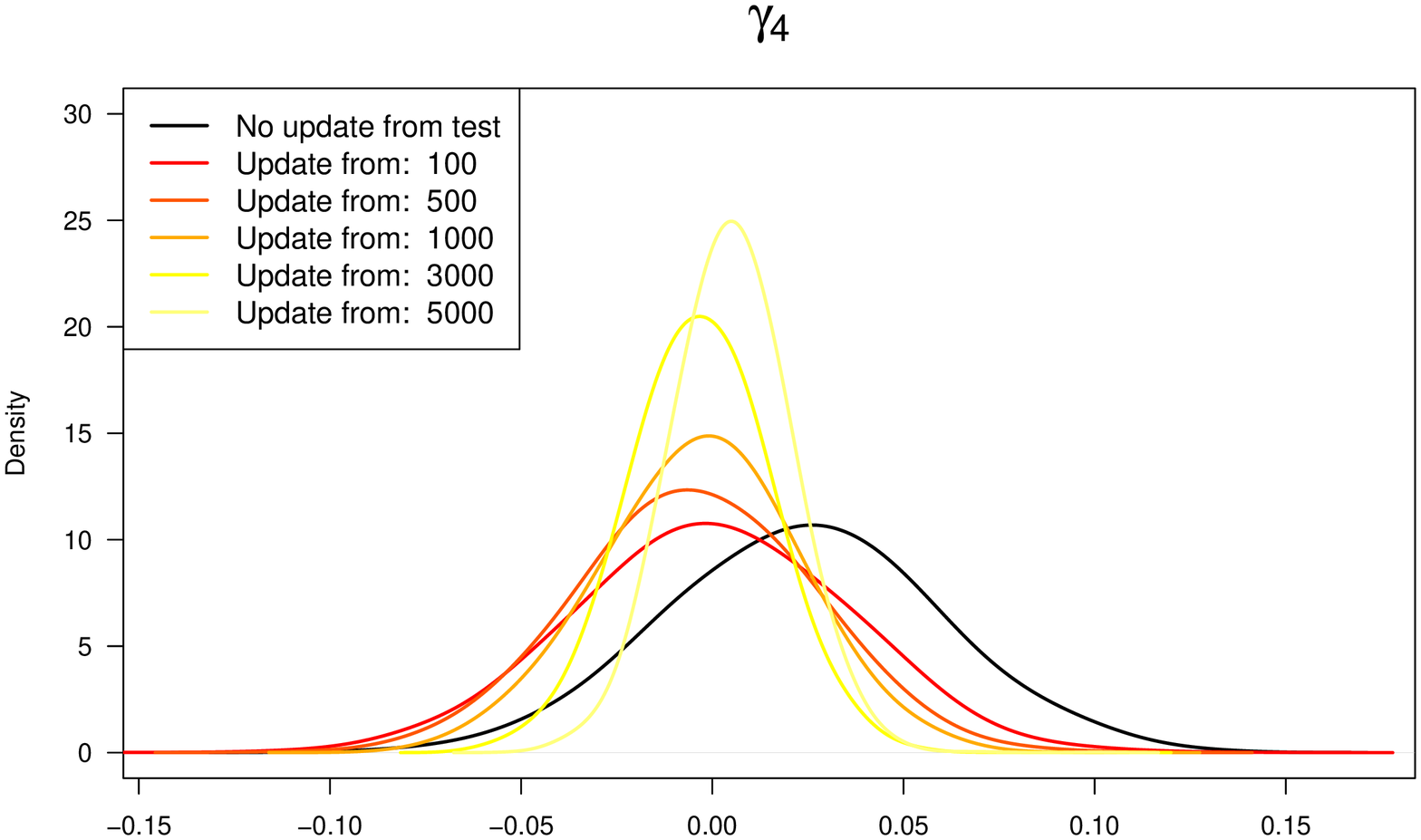}
\end{subfigure} 
\hfill
\begin{subfigure}{.45\linewidth}
  \centering
    \includegraphics[width=\textwidth]{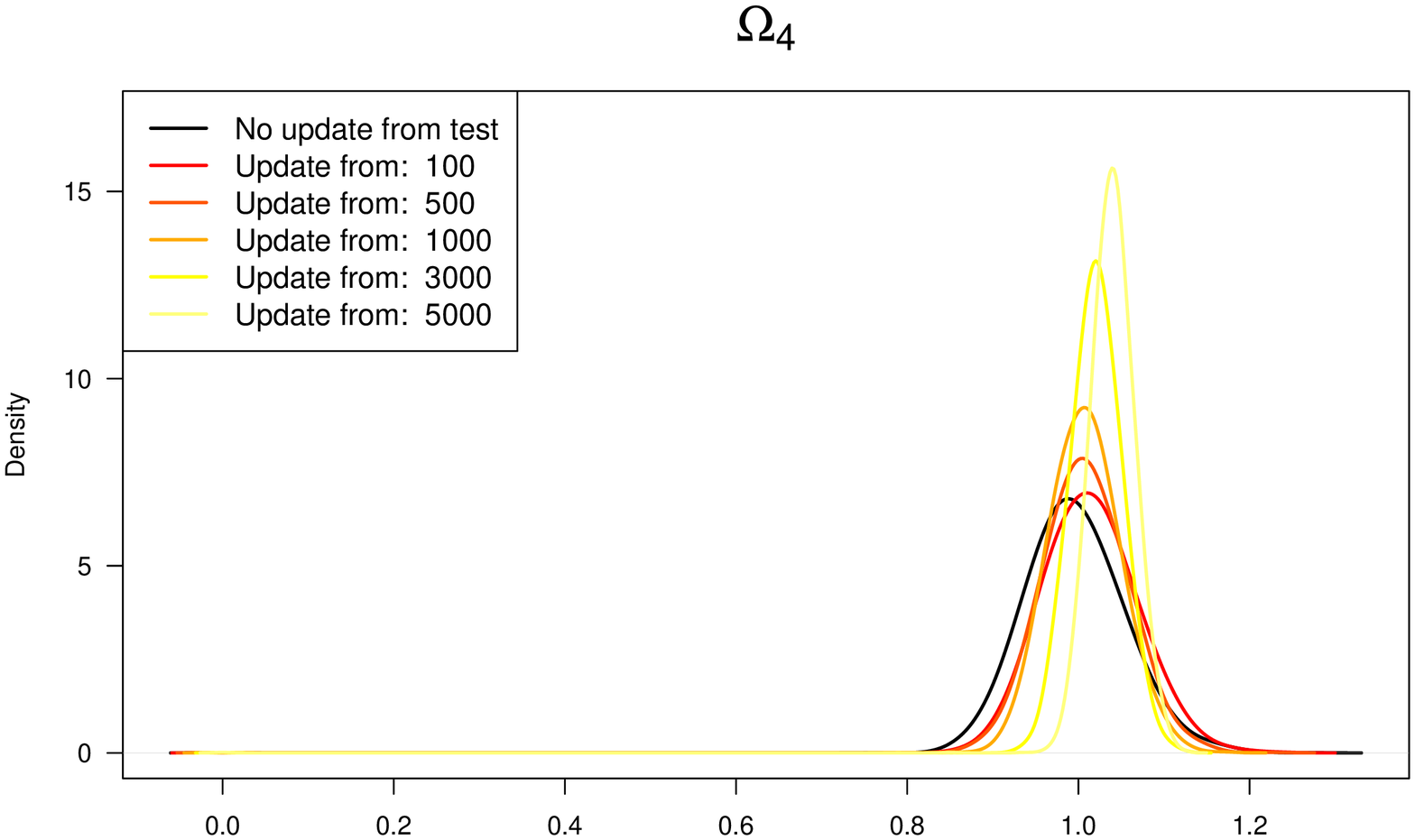}
\end{subfigure} 

\begin{subfigure}{0.45\linewidth}
  \centering
    \includegraphics[width=\textwidth]{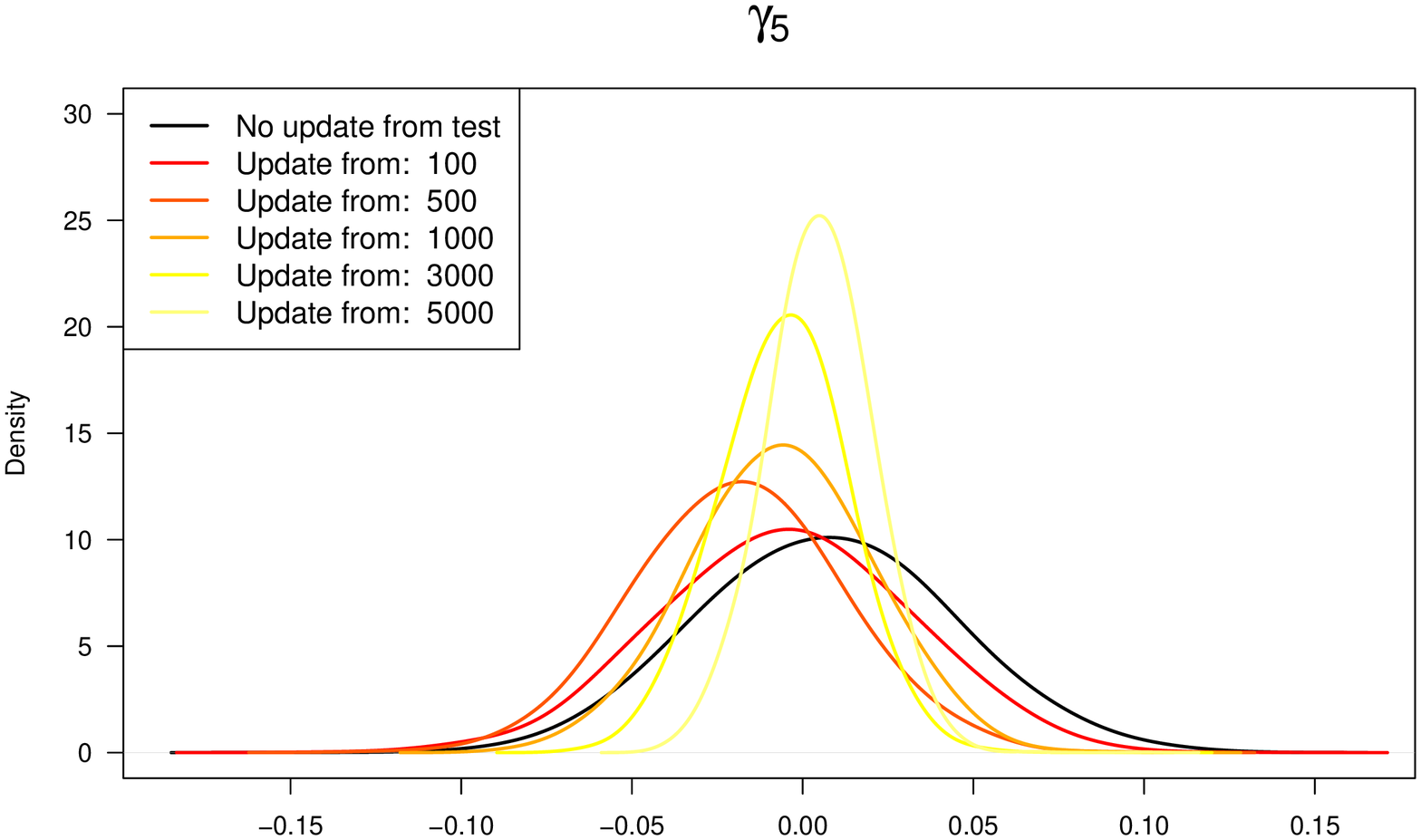}
\end{subfigure} 
\hfill
\begin{subfigure}{.45\linewidth}
  \centering
    \includegraphics[width=\textwidth]{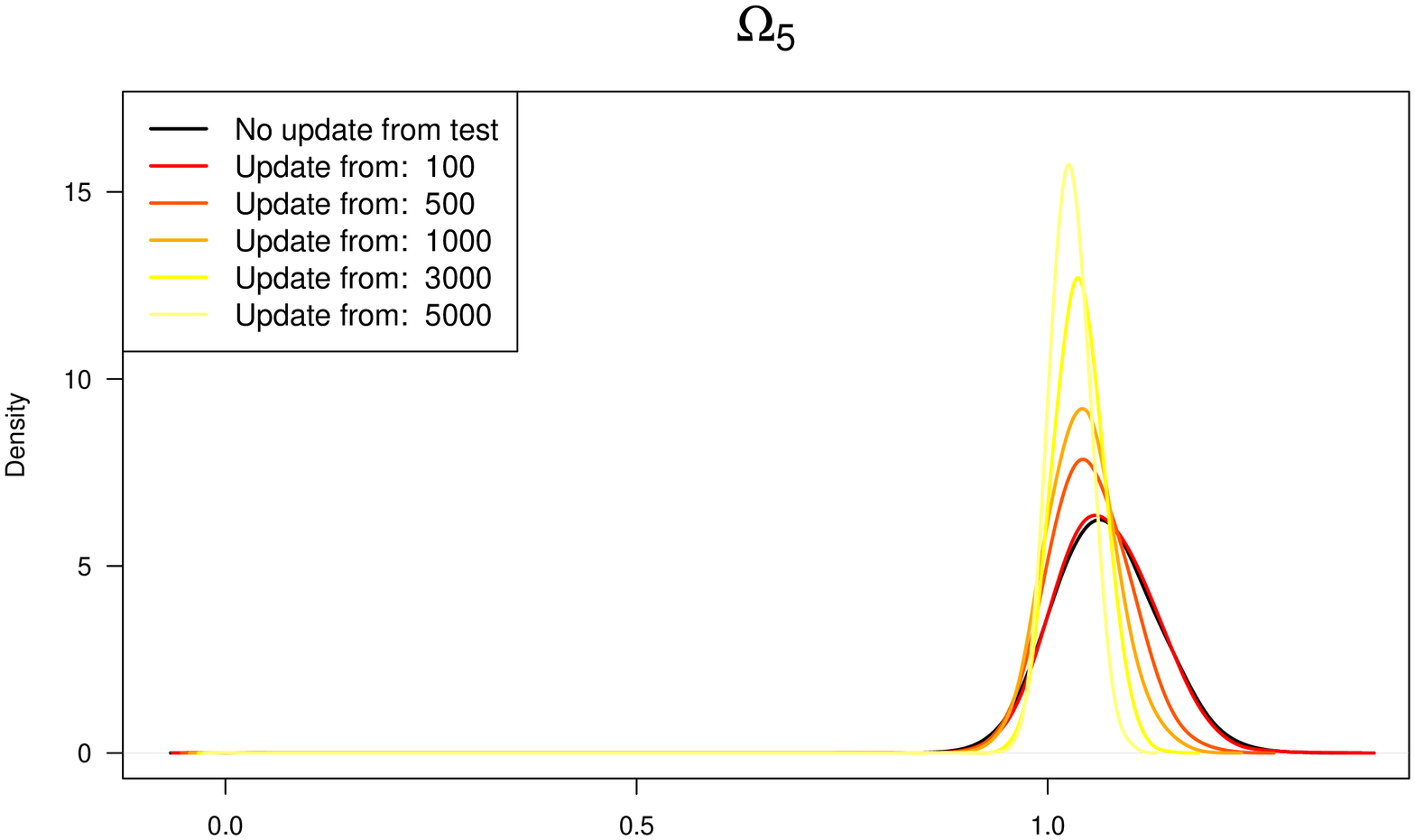}
\end{subfigure} 
\caption{The effect on the posteriors for the parameters in the missing value model, $\v \gamma$ (left column) and $\v \Omega$ (right column), as increasingly more test data is observed. The plots in each row correspond to one of the three missing values in the observation of interest.}\label{fig:dirtystrategymissingparam}
\end{figure}

\begin{figure}[!h]
\includegraphics[width=0.7\textwidth]{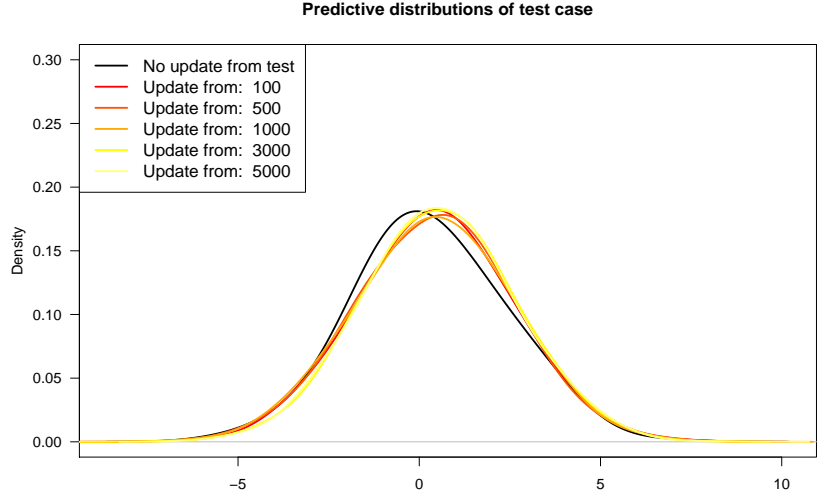}
\caption{Predictive densities for the observation as increasingly more test data is observed.}
\label{fig:test_dens_pred}
\end{figure}

\section{Simulations}\label{sec:simulations}

The simulated data used in the following subsections are of dimensions $p=40$, $n=1000$ and $\tilde{n}=1000$. The covariates are divided into two independent groups with sizes $25$ and $15$ respectively, and the correlation between any pair of covariates within each group is $\rho = 0.8$. Half of the covariates are significant in the data generating model and the $R^2$ of the regression ranges from $0.65$ to $0.80$. More details on the data simulation process is given in Appendix \ref{Appendix:data}. The data are censored according to the following  principle, which aims to mimick the censoring due to interference from the strongest signal among a set of signals (covariates):
\begin{equation}\label{eq:cens}
    x_{ij} = \begin{cases} x_{ij} & \text{if } x_{ij} \geq \text{max}(\xobsi)-\Delta \\
                            \text{max}(\xobsi)-\Delta & \text{otherwise},
    \end{cases}
\end{equation}

\noindent where $\Delta$ is a known distance from the strongest signal for which covariates of lower amplitude are still detectable. We compare our Bayesian imputation to two baselines: i) an idealized model using the uncensored (complete) data and ii) a model with missing covariates imputed with a commonly used na\"{\i}ve imputation strategy where
\begin{equation*}
    \xmisi = \text{max}(\xobsi)-\Delta,
\end{equation*}
meaning that all missing values are imputed to the lowest limit of detection.

\subsection{Assessing predictive performance on artificial data}

In this subsection, we evaluate the performance of our model on artificial data of high dimension subject to approximately $40$\% censoring; the exact censoring level varies somewhat due to the censoring mechanism in \eqref{eq:cens}. No auxillary variables are used in this experiment, the effect of such variables on the predictive performance is explored separately in the next subsection. 

\begin{figure}[!h]
\includegraphics[width=0.7\textwidth]{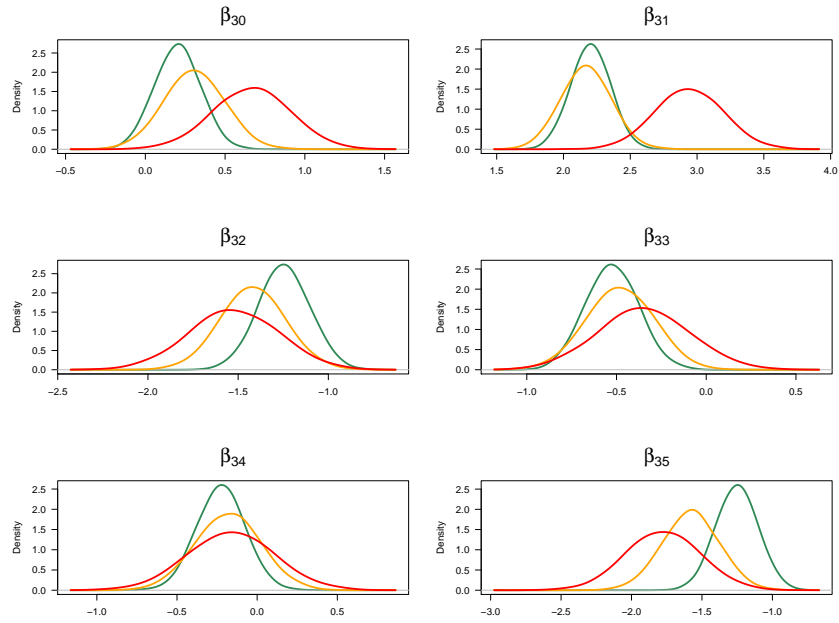}
\caption{Posterior densities of some non-zero $\v \beta$ coefficients for the artificial data. The green density is from the complete data model, the orange density is from Bayesian imputation, and the red density is from the na\"{\i}ve imputation.}
\label{fig:gauss_40_beta}
\end{figure}

\begin{figure}[!h]
\includegraphics[width=0.7\textwidth]{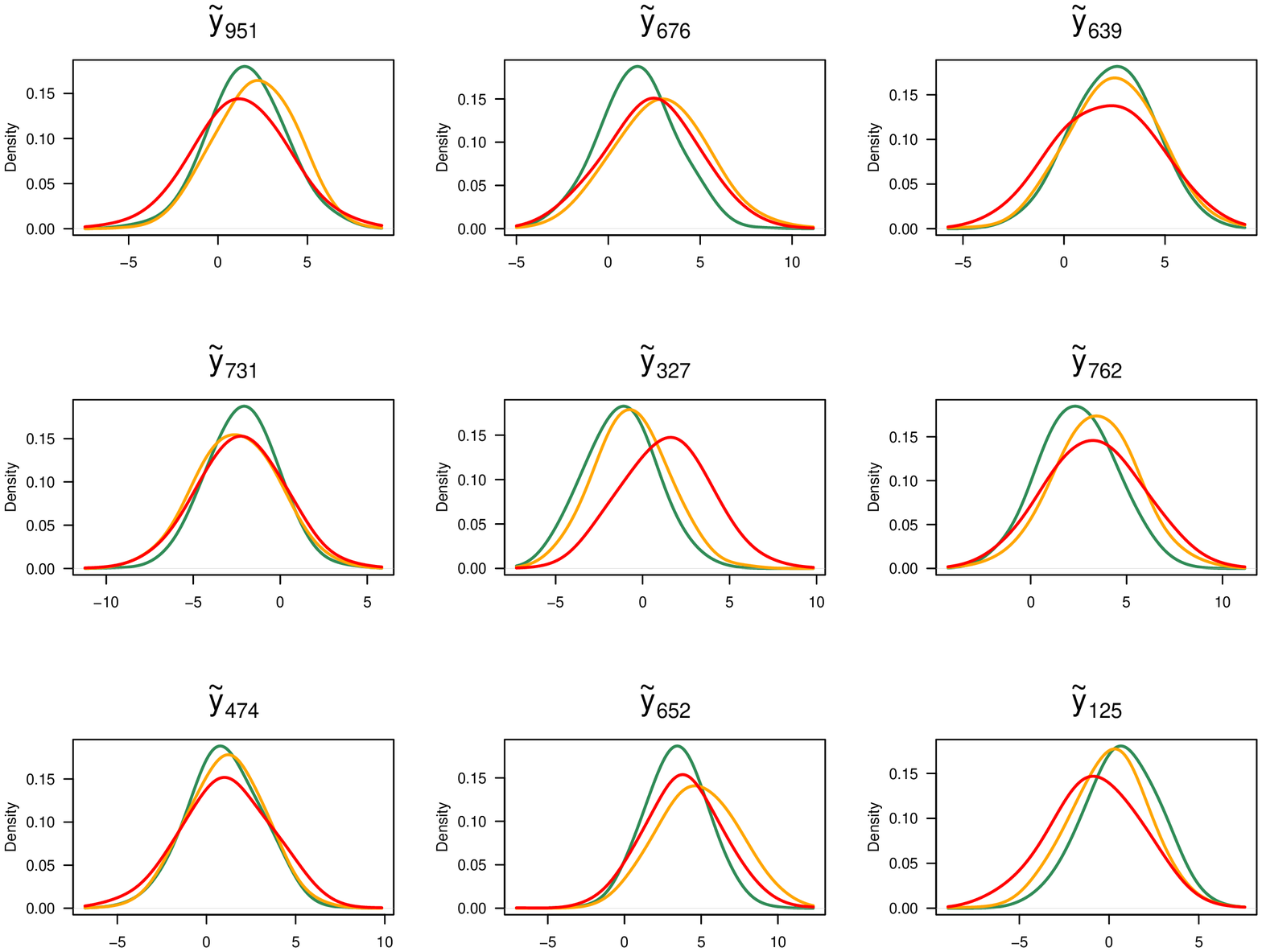}
\caption{Predictive distribution densities for the artificial data.}
\label{fig:gauss_40_preds}
\end{figure}

Figure \ref{fig:gauss_40_beta} illustrates the posterior densities of $\v \beta$ for a subset of the significant covariates for one of the generated datasets. Our Bayesian imputation method produces densities closer to the complete data model compared to the na\"{\i}ve imputation. Furthermore, the prior variance is lower and closer to the complete data model for all $\v \beta$ coefficients. Similarly, Figure \ref{fig:gauss_40_preds} shows that our Bayesian imputation gives predictive distributions that are in general closer to the complete data model than the na\"{\i}ve model, which have less certainty and/or is shifted away from the complete data model densities.

To more formally assess the predictive performance over all datasets and observations, we use the log predictive score measure on the test data \citep{gelman1995bayesian} 
\begin{equation*}
    \sum_{i=1}^{\tilde n} \log p(\tilde y_i \vert \xobsi, \v y, \v X^{(\mathrm{o})}).
\end{equation*}
Table \ref{tab:sum_scores_gauss} shows the sum of the log predictive scores over each of the five sets of test data. The log predictive scores are consistently higher for the Bayesian imputation model compared to the na\"{\i}ve model. Figure \ref{fig:gauss_40_scores} present an alternative view of the prediction performance by kernel density plots of the log predictive density evaluations $\log p(y_i \vert \xobsi)$ over all $\tilde n = 1000$ test observations in Dataset 1 in Table \ref{tab:sum_scores_gauss}; the four other datasets gave similar results.

\begin{table}[h!]
    \centering
\begin{tabular}{lrrrrr|r}
    Imputation               & \text{Dataset} 1 & \text{Dataset} 2 & \text{Dataset} 3 &  \text{Dataset} 4 & \text{Dataset} 5 & Average \\ \toprule
    Complete            & -2179 & -2156 & -2156 & -2130 & -2131 & -2150 \\
    Bayesian &  -2316 & -2249 & -2299 & -2256 & -2307 &  -2285 \\
    Na\"{\i}ve    &  -2388 & -2369 & -2449 & -2323 & -2369 & -2380 \\
\bottomrule
\end{tabular}
    \caption{Sum of log predictive test scores for the artificial data.}
    \label{tab:sum_scores_gauss}
\end{table}

\begin{figure}[!h]
\includegraphics[width=0.7\textwidth]{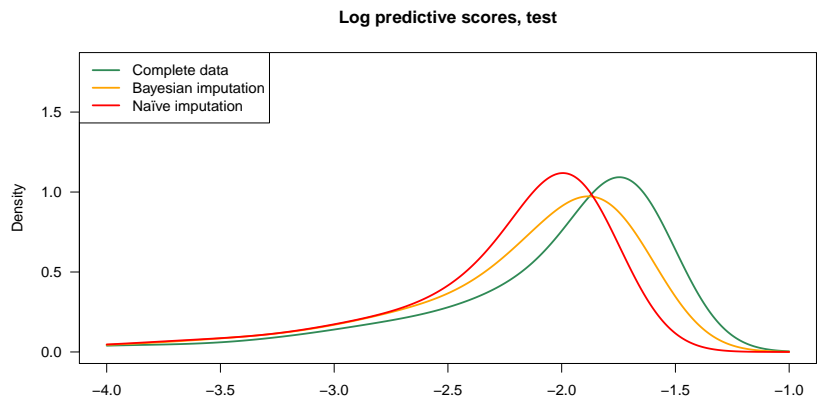}
\caption{Kernel density estimates of the log predictive density values for the $\tilde{n} = 1000$ test cases in Dataset 1.}
\label{fig:gauss_40_scores}
\end{figure}

\subsection{Auxiliary variables for improved prediction performance}

In this subsection, the multivariate regression model for $\v X$ in \eqref{eq:model} is included and the $q = 10$ auxiliary variables in $\v w$ explain between $65$ and $90$ percent of the variation in the $p = 5$ covariates in $\v x$. The covariates in the predictive model for $\v y$ are uncorrelated in this scenario conditional on $\v w$.

Table \ref{tab:sum_scores_aux} shows that the model with auxiliary variables dramatically outperforms the model without auxiliary variables for all five simulated data sets. The log predictive density values for the $\tilde n = 1000$ test cases in Dataset 1 are displayed as kernel density plots in Figure \ref{fig:scores_aux}. Having access to informative auxiliary variables are clearly useful for prediction. This is important as industry agents can often easily collect variables that may be useful as auxiliary variables, for example aspects of the connecting devices in the telecom  application.

\begin{table}[h!]
    \centering
\begin{tabular}{lrrrrr|r}
    Imputation                                       & Dataset 1 & Dataset 2 & Dataset 3 & Dataset 4 & Dataset 5 & Average \\ \toprule
    Complete                                    & -2109 & -2100 & -2094 & -2121 & -2104 & -2106 \\
    
    Bayesian, with $\v W$        & -2221 & -2193 & -2151 & -2196 & -2126 & -2178 \\
    Bayesian, no $\v W$     & -2406 & -2439 & -2371 & -2370 & -2275 & -2372 \\\bottomrule
\end{tabular}
    \caption{Log predictive score comparison for the model with and without auxiliary variables.}
    \label{tab:sum_scores_aux}
\end{table}

\begin{figure}[!ht]
\includegraphics[width=0.7\textwidth]{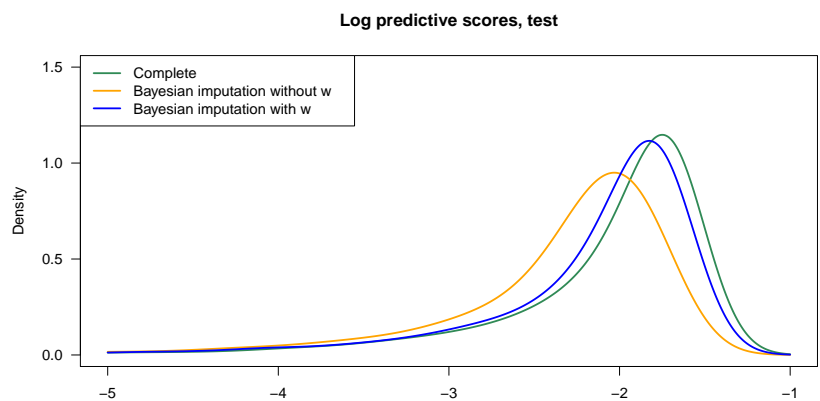}
\caption{Kernel density estimates of the log predictive density values for the $\tilde{n} = 5000$ test cases in Dataset 1-5, with and without auxiliary variables.}
\label{fig:scores_aux}
\end{figure}

\section{Application to signal strength censoring}\label{sec:application}

High dimensional censored data are standard in wireless communications networks. There are usually several radio carrier frequencies available for connection at each location, with each user typically being connected to one of them. To optimize the user connection, the network should ideally make autonomous decisions regarding carrier frequency. When deciding whether or not to switch frequencies for a user, the user equipment disconnects briefly to perform a signal strength measurement on alternative frequencies. To avoid the need for disconnection, it is of great interest to predict the signal strengths on alternative frequencies instead of disconnecting to measure them.

We use data from a sophisticated simulator of signal strengths in a wireless network from one of the world's largest telecommunication companies, Ericsson AB. The covariates are signal strengths from the carrier frequency \textit{cells}, which are smaller geographical areas within the wireless network from which a connection can be made. An example is depicted in Figure \ref{fig:network}.  At each location, the user is assigned to the cell which offers the most reliable connection on that carrier frequency. The response variable in the regression is the maximum signal strength available at an alternative frequency. Cells which carry the same signal frequency may interfere with each other, causing signals to be censored if not strong enough compared to the strongest signal \citep{3gpp}. For these data, $p = 37$, $n = 800$ and $\tilde{n} = 200$.

\begin{figure}[!h]
\includegraphics[width=0.4\textwidth]{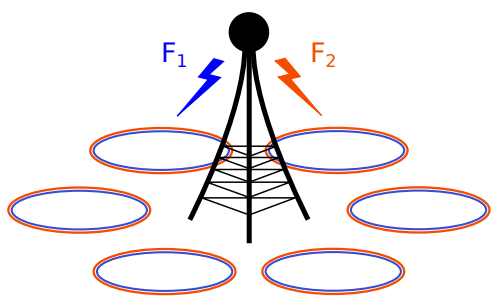}
\caption{A simple example of a base station transmitting two frequencies to surrounding cells.}
\label{fig:network}
\end{figure}

\begin{figure}[!h]
\includegraphics[width=1\textwidth]{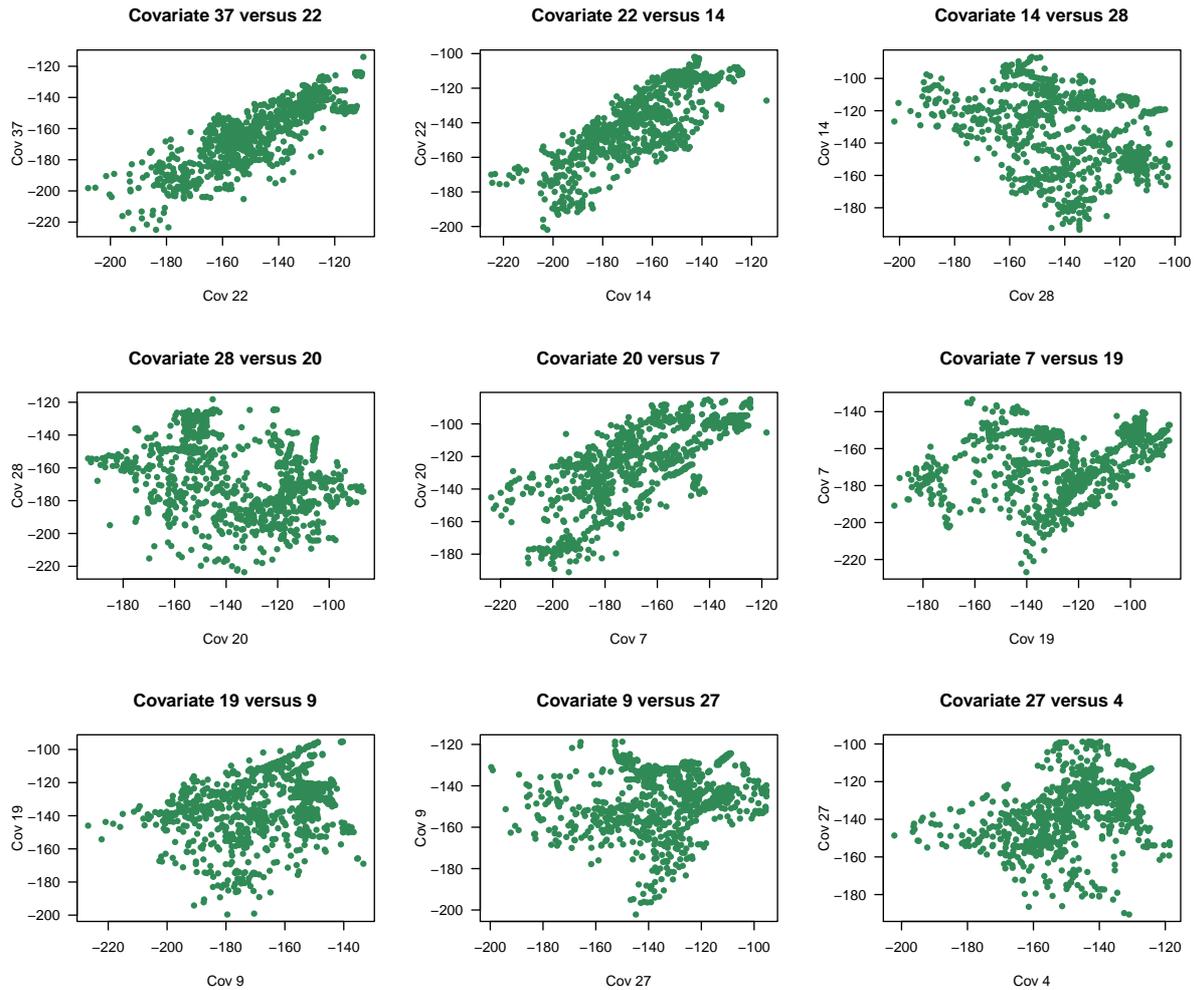}
\caption{Scatter plots for randomly selected pairs of covariates in the original wireless network simulator data.}
\label{fig:realCovDensities}
\end{figure}

The original simulator data have a spatial component that is not available in the dataset used here. This lacking spatial information causes the covariates to be highly non-Gaussian, as is illustrated in Figure \ref{fig:realCovDensities}. We therefore instead generate new artificial data that mimic the dataset from Ericsson's simulator. This is achieved by first estimating the model parameters from the complete data from the wireless network simulator. The posterior mean of the model parameters are then used for simulating five different data sets, each with $n = 1000$ training observations and $\tilde n = 1000$ test observations. This way, these datasets reflect Ericsson's situation, but do not violate the Gaussian assumption of the model. We will nevertheless refer to these simulated data as the \emph{wireless network data}, to distinguish them from the previously presented artificial datasets. We set the detection limit $\Delta$ so that we get approximately $25\%$ censored data on average over the observations.

Figures \ref{fig:sim50Beta} and \ref{fig:sim50Y} display a subset of $\v \beta$ posterior densities and predictive densities for some randomly selected test cases. The posteriors for $\v \beta$ under the Bayesian imputation are generally closer to the densities from complete data when compared to the densities from a na\"{\i}ve imputation. For most of the predictive distributions the Bayesian imputation achieves a density similar to the complete densities, while the na\"{i}ve imputation densities are shifted away from the complete data density for the majority of the test cases. Table \ref{tab:sum_scores_sim_real} and Figure \ref{fig:sim50Scores} show that the Bayesian imputation generally attains substantially higher log predictive scores than the na\"{i}ve strategy. For all the simulated data sets, the Bayesian imputation model achieves a higher sum of log predictive scores than the na\"{i}ve imputation model.

\begin{figure}[!h]
\includegraphics[width=0.7\textwidth]{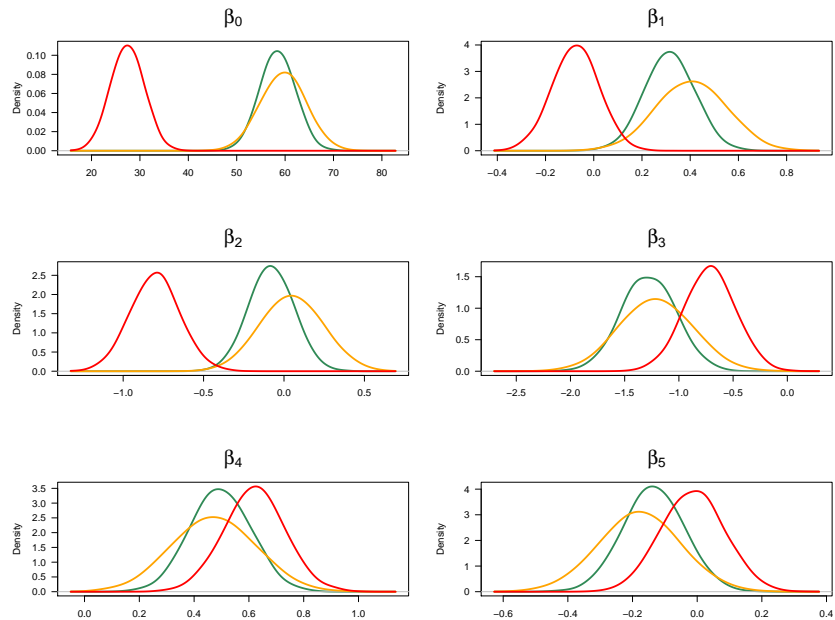}
\caption{Posterior densities for the first six regression coefficients in Dataset 1 for the wireless network data.}
\label{fig:sim50Beta}
\end{figure}

\begin{figure}[!h]
\includegraphics[width=0.7\textwidth]{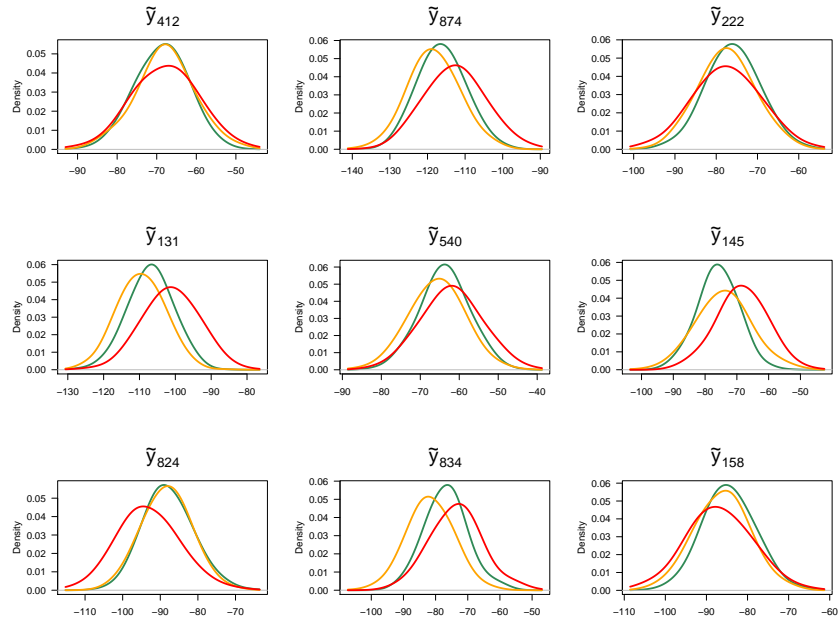}
\caption{Predictive densities for some randomly selected test cases in Dataset 1 for the wireless network data.}
\label{fig:sim50Y}
\end{figure}

\begin{table}[h!]
    \centering
\begin{tabular}{lrrrrr|r}\toprule
    Imputation               & Dataset 1 & Dataset 2 & Dataset 3 & Dataset 4 & Dataset 5 & Average \\ \toprule
    Complete            & -3277 & -3285 & -3265 & -3312 & -3228 & -3273 \\
    
    Bayesian & -3374 & -3407 & -3364 & -3408 & -3337 & -3378 \\
    Na\"{\i}ve    & -3486 & -3503 & -3461 & -3511 & -3470 & -3486 \\\bottomrule
\end{tabular}
    \caption{Log predictive scores for the wireless network data.}
    \label{tab:sum_scores_sim_real}
\end{table}

\begin{figure}[!h]
\includegraphics[width=0.7\textwidth]{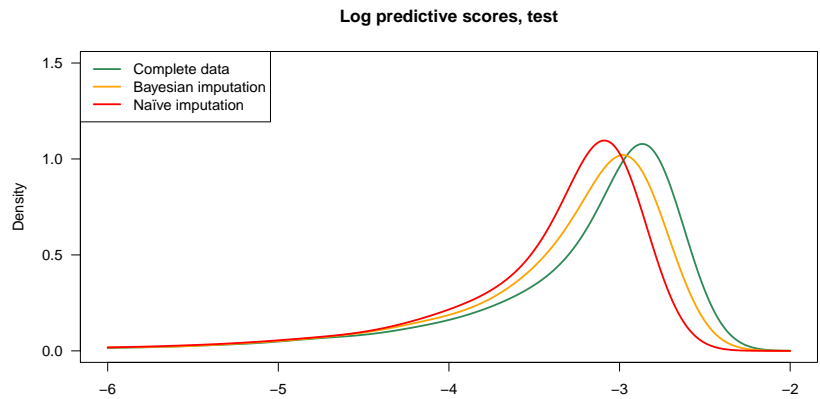}
\caption{Kernel density estimates of the log predictive density values for the $\tilde{n} = 5000$ test cases in Dataset 1-5 from the wireless network data.}
\label{fig:sim50Scores}
\end{figure}

\section{Conclusions}\label{sec:discussion}

We have presented an efficient Gibbs sampling algorithm for regression or classification with missing covariate observations. The algorithm samples the missing values jointly, which can be at least two orders of magnitude as efficient as univariate sampling of missing values. The conditions for this improved efficiency is clarified by deriving the posterior correlation of the missing values. A scheme for simulating from the predictive distribution is proposed and the extra complication from having to re-run the Gibbs sampler on the training data in the prediction phase is highlighted. The predictive performance of the model is documented on artificial data and on data from the telecom sector.

The efficiency of the proposed joint updates of the missing values depends on the recently developed simulation algorithm for truncated multivariate normal distributions in \citet{botev2017normal}. Botev's algorithm uses a highly efficient rejection sampling based on a very accurate minimax tilting method to approximate the posterior density, which makes it reliable and efficient also in high dimensions. While we have found it to work flawlessly for the simulation setups in this paper, it could of course break down in extreme dimensions. For such cases we instead proposed to sample the missing values in a smaller number of blocks, using Lemma \ref{lem:jointsampling} to exploit the correlation structure when arranging the missing values in blocks in the most efficient way.

Efficiency for a given time budget can be further enhanced by using Random scan updates of the missing covariates where only a random subset of observations are updated in each Gibbs iteration. Random scan is shown to be particularly useful when the aim is prediction.

We use artificially simulated data to show that our Bayesian imputation method is consistently superior to the na\"{\i}ve imputation model when evaluating predictive density performance in terms of log predictive scores. 

Two strategies for making predictions for test cases are proposed: an exact yet computationally costly way and a more affordable approximation that takes the shortcut of not updating the model parameters conditioned on new test cases. The approximation is demonstrated to be give predictive distributions that are quite close to the ones from the exact strategy, at least when the test sets are not too large in relation to the size of the training data. This can be ensured by retraining the model regularly after an appropriate amount of test updates.

Our model allow the use of auxiliary variables in a multivariate regression model for the missing covariates. We show that the inference for the missing values can efficiently exploit the auxiliary variables and dramatically outperform the model without auxiliary variables in predictive performance.

We assume that the detection limit, as determined by $\Delta$, is known. However, this is not always the case for interference problems, and future work should put efforts to generalize the framework to unknown detection limits. This can be straightforward achieved in principle by adding a Metropolis-Hastings updating step for $\Delta$ to the Gibbs sampler, but issues of parameter identification should be explored in detail.

The method relies on the covariates being Gaussian, and it would be interesting to extend the modeling structure to more complex data. One straightforward approach is to assume a multivariate Gaussian mixture model for the covariates, a model which is well known to be amenable to Gibbs sampling by simply augmenting the model latent mixture allocation indicators for each observation \citep{bishop2006pattern}. 

Finally, the presented framework is based on linear regression or classification models like probit or logistic regression with linear decision boundaries. The extension to non-linear models is not straightforward, not even for models like polynomials which remain linear in the regression coefficients, since the full conditional posteriors for the missing values become intractable. Future developments should therefore develop efficient proposal distributions for Metropolis-Hastings updates of the missing values.

\bibliographystyle{apalike}
\bibliography{DetectionLimit}

\appendix

\section{Proofs}\label{Appendix:Proofs}

\subsection{Proof of Lemma \ref{lem:jointsampling}}

\begin{proof}
Using a version of the Sherman-Morrison formula
$$(\boldsymbol{aa^{\top}}+\boldsymbol{A})^{-1}=\boldsymbol{A}^{-1}-\boldsymbol{A}^{-1}\boldsymbol{a}\boldsymbol{a}^{\top}\boldsymbol{A}^{-1}(1+\boldsymbol{a^{\top}}\boldsymbol{A}^{-1}\boldsymbol{a})^{-1}$$
for a $p$-dimensional vector $\boldsymbol{a}$ and invertible $p\times p$
matrix $\boldsymbol{A}$ \citep{harville1998matrix}[Corollary 18.2.10] the full
conditional posterior covariance of the missing values $\xmis$ given in Section \ref{sec:inference} can
be written
\end{proof}
\begin{align}
\mathrm{Cov}(\xmis|\xobs,y) & =\Bigg(\frac{1}{\sigma^2}\v \beta_\mathrm{m} \v \beta_\mathrm{m}^\top + \bar{\v \Sigma}^{-1}\Bigg)^{-1}\nonumber \\
 & =\bar{\v \Sigma}-\bar{\v \Sigma}\frac{1}{\sigma_{y}}\boldsymbol{\beta}_\mathrm{m}\frac{1}{\sigma_{y}}\boldsymbol{\beta}_\mathrm{m}^{\top}\bar{\v \Sigma}\left(1+\frac{1}{\sigma_{y}^{2}}\boldsymbol{\beta}_\mathrm{m}^{\top}\bar{\v \Sigma}\boldsymbol{\beta}_\mathrm{m}\right)^{-1}\nonumber \\
 & =\bar{\v \Sigma}-\frac{\bar{\v \Sigma}\boldsymbol{\beta}_\mathrm{m}\boldsymbol{\beta}_\mathrm{m}^{\top}\bar{\v \Sigma}}{\sigma_{y}^{2}+\boldsymbol{\beta}_\mathrm{m}^{\top}\bar{\v \Sigma}\boldsymbol{\beta}_\mathrm{m}}.\label{eq:CovPostInitialExpression}
\end{align}
Note that \-
\[
\mathrm{Cov}(\xmis,y|\xobs)=\mathrm{Cov}(\xmis,\v \beta_\mathrm{m}^{\top}\xmis+\varepsilon)=\bar{\v \Sigma}\v \beta_\mathrm{m}
\]
and 
\[
\mathrm{Var}(y|\xobs)=\mathrm{Var}(\v \beta_\mathrm{m}^{\top}\xmis+\varepsilon|\xobs)=\mathrm{Var}(\v \beta_\mathrm{m}^{\top}\xmis+\varepsilon|\xobs)=\v \beta_\mathrm{m}^{\top}\bar{\v \Sigma}\boldsymbol{\beta}_\mathrm{m}+\sigma_{y}^{2},
\]
and by definition $\bar{\v \Sigma}=\mathrm{Var}(\xmis\vert\xobs)$.
We can therefore write \eqref{eq:CovPostInitialExpression} as

\begin{align}
\mathrm{Cov}(\xmis|\xobs,y) & =\mathrm{Cov}(\xmis\vert\xobs)-\frac{\mathrm{Cov}(\xmis,y|\xobs)\mathrm{Cov}(\xmis,y|\xobs)^{\top}}{\mathrm{Var}(y|\xobs)}\nonumber \\
 & =\mathrm{Cov}(\xmis\vert\xobs)-\frac{\v S \v S^{-1}\mathrm{Cov}(\xmis,y|\xobs)\mathrm{Cov}(\xmis,y|\xobs)^{\top}\v S^{-1}\v S}{\sqrt{\mathrm{Var}(y|\xobs)}\sqrt{\mathrm{Var}(y|\xobs)}}\label{eq:CovPost}\\
 & =\mathrm{Cov}(\xmis\vert\xobs)-\v S\rho(\xmis,y|\xobs)\rho(\xmis,y|\xobs){}^{\top}\v S,\nonumber 
\end{align}
where $\v S=\mathrm{Diag}\left(\sqrt{\mathrm{Var}(\xmis_{k}\vert\xobs)}\right)$.
The $k$th diagonal element of $\mathrm{Cov}(\xmis|\xobs,y)$
can then be read of \eqref{eq:CovPost} as 
\[
\mathrm{Var}(\xmis_{k}|\xobs,y)=\mathrm{Var}(\xmis_{k}\vert\xobs)\left(1-\rho^{2}(\xmis_{k},y|\xobs)\right).
\]
The expression for the posterior correlation matrix $\rho(\xmis|\xobs,y)$
in the lemma is finally obtained by simplifying 
\[
\rho(\xmis|\xobs,y)=\tilde{\boldsymbol{S}}^{-1}\mathrm{Cov}(\xmis|\xobs,y)\tilde{\v S}^{-1},
\]
where $\mathrm{\tilde{\v S}=Diag}\left(\sqrt{\mathrm{Var}(\xmis_{k}|\xobs,y)}\right)$.

\subsection{Derivation of the full conditional posteriors}

\subsubsection*{Full conditional for $\v\beta$ and $\sigma^2$}

Since the full conditional posterior of $\v\beta$ and $\sigma^2$ conditions on the complete data, these posteriors follow from standard results on Bayesian linear regression, see e.g. \citet{bishop2006pattern}.

\subsubsection*{Full conditional for $\v \Gamma$ and $\v \Omega$}

Conditional on all other model parameters, the likelihood part in the joint posterior of $\v \Gamma$ and $\v \Omega$ is given by the likelihood of a multivariate regression model
\begin{equation*}
    \v X = \v W \v \Gamma + \v E, 
\end{equation*}
where the rows of $\v E$ are iid from $N(\v 0,\v \Omega)$. The likelihood for this model is
\begin{align*}
    p(\v X| \v W, \v \Gamma, \v \Omega) = & \prod_{i=1}^n{|2\pi\v \Omega|^{-1/2}\text{ exp}
    \Bigg(-\frac{1}{2}\big(\v x_i-\v \Gamma^{\top} \v w_i\big)^{\top} \v \Omega^{-1}\big(\v x_i-\v \Gamma^{\top}\v w_i\big)\Bigg)} \\
    = & |2\pi\v \Omega|^{-n/2}\text{ exp}
    \Bigg(-\frac{1}{2}\text{tr}\v \Omega^{-1}\big(\v X-\v W \v \Gamma)^{\top} \big(\v X-\v W \v \Gamma)\Bigg)\\
    = & |2\pi\v \Omega|^{-n/2}\text{ exp}
    \Bigg(-\frac{1}{2}\text{tr}\v \Omega^{-1}\big(n \v S + \big(\v \Gamma - \hat{\v \Gamma}\big)^\top \v W^\top \v W \big(\v \Gamma - \hat{\v \Gamma}\big)\Bigg),
\end{align*}
where 
\begin{align*}
    \hat{\v \Gamma} = & \big( \v W^\top \v W \big) ^{-1} \v W^\top \v X,\\
    \v S = & \big(\v X - \v W \hat{\v \Gamma} \big)^\top \big(\v X - \v W \hat{\v \Gamma} \big)/n.
\end{align*}
Using the matrix identity \citep{harville1998matrix}
\begin{equation*}
    \text{tr }\big(\v A_1^\top \v A_2 \v A_3 \v A_4^\top \big) = \big(\operatorname{vec} \v A_1\big)^\top\big(\v A_4 \otimes \v A_2 \big) \operatorname{vec} \v A_3
\end{equation*}
with
\begin{equation*}
    \v A_1 =  \v \Gamma - \hat{\v \Gamma},
    \v A_2 =  \v W^\top \v W,
    \v A_3 =  \v \Gamma - \hat{\v \Gamma}, \text{ and }
    \v A_4 =  \v \Omega^{-1},
\end{equation*}
we get the likelihood
\begin{equation*}
    p(\v X| \v W,\v \Gamma, \v \Omega) = |2 \pi \v \Omega|^{-n/2} \text{exp}\Bigg(- \frac{1}{2} \text{tr }\v \Omega^{-1} n \v S \Bigg) \text{exp}\Bigg(\frac{1}{2} \operatorname{tr} \big( \v \gamma - \hat{\v \gamma} \big)^\top \big(\v \Omega^{-1} \otimes \v W^\top \v W \big) \big( \v \gamma - \hat{\v \gamma} \big) \Bigg).
\end{equation*}
Multiplying this likelihood with the prior
\begin{align*}
    \gamma|\v \Omega & \sim \mathcal{N}(0, \v \Omega \otimes \tau^2_{\v \gamma}\v I_r) \\
    \v \Omega & \sim IW(\v A, \kappa),
\end{align*}
and completing the squares in the exponents results in the joint posterior density in Section \ref{sec:inference}.

\subsubsection*{Full conditional for $\xmis_{i}$}

The full conditional posterior for $\xmis_{i}$ is
\begin{align}\label{eq:postmissing}
    p(\xmisi | y_i, \xobsi, \v W, \v \beta, \sigma, \v \Gamma, \v \Omega) & \propto p(y_i | \xmisi, \xobsi, \v \beta, \sigma)\\
    &  \times p(\xmisi | \xobsi, \v W, \v \Gamma, \v \Omega) \cdot I(\xmisi \leq \v c_i) \nonumber, 
\end{align}
where $I(\xmisi \leq \v c_i)=1$ if all elements in $\xmisi$ are smaller than their corresponding detection limits in the vector $\v c_i$, otherwise $I(\xmisi \leq \v c_i)=0$.

\noindent Partition the joint distribution of the missing and observed covariates as
\begin{equation*}
    \begin{pmatrix}
                \xmisi  \\
                \xobsi  
    \end{pmatrix}
    \sim \mathcal{N}\Bigg[   
    \begin{pmatrix}
                \xmishat_i  \\
                \xobshat_i  
    \end{pmatrix},
    \begin{pmatrix}
        \v \Omega_{\mathrm{m}\mathrm{m}} & \v \Omega_{\mathrm{m}\mathrm{o}}\\
        \v \Omega_{\mathrm{o}\mathrm{m}} & \v \Omega_{\mathrm{o}\mathrm{o}}
    \end{pmatrix}
\Bigg],
\end{equation*}
where $\hat{\v x}_i = (\xmishat_i{}^{\top},\xobshat_i{}^{\top})^\top = \v \Gamma^\top \v w_i$. The factor $p(\xmisi | \xobsi, \v W, \v \Gamma, \v \Omega)$ can now be explicitly expressed using the conditioning properties of the multivariate normal distribution \citep{harville1998matrix} as 
\begin{equation}\label{eq:conditioningresult}
      \xmisi | \xobsi, \v W, \v \Gamma, \v \Omega \sim \mathcal{N}(\bar{\v \mu}_i, \bar{\v \Sigma}_i), 
\end{equation}
where
\begin{align*}
\bar{\v \mu}_i &= \xmishat_i +\v \Omega_{\mathrm{m}\mathrm{o}}\v \Omega_{\mathrm{o}\mathrm{o}}^{-1}(\xobsi-\xobshat_i) \\
\bar{\v \Sigma}_i &= \v \Omega_{\mathrm{m}\mathrm{m}} - \v \Omega_{\mathrm{m}\mathrm{o}}\v \Omega_{\mathrm{o}\mathrm{o}}^{-1}\v \Omega_{\mathrm{o}\mathrm{m}}.
\end{align*}
Now, using \ref{eq:conditioningresult} in \ref{eq:postmissing} we get
\begin{align*}
        p(\xmisi &| y_i, \xobsi, \v W, \v \beta, \sigma, \v \Gamma, \v \Omega)  \propto p(y_i | \xmisi, \xobsi, \v \beta, \sigma) \cdot p(\xmisi | \xobsi, \v W, \v \Gamma, \v \Omega) \cdot I(\xmisi \leq \v c_i)\\ \nonumber
        & \propto  \exp \Bigg\{ \sigma^{-2} \Big(\tilde{y}_i- \v \beta_\mathrm{m}^\top \xmisi \Big)^2 
            + \Big(\xmisi-\bar{\v \mu}_i\Big)^\top \v \Sigma_i^{-1}  \Big(\xmisi-\bar{\v \mu}_i\Big)\Bigg\}\cdot I(\xmisi \leq \v c_i),
\end{align*}
where $\tilde{y}_i = y_i - \beta_0 - \v \beta_\mathrm{o}^\top \xobsi$. Completing the square in the exponent gives
\begin{equation*}
    p(\xmisi | y_i, \xobsi, \v W, \v \beta, \sigma, \v \Gamma, \v \Omega)  \propto 
    \exp \Bigg\{ -\frac{1}{2}\big(\xmisi-\boldsymbol{\mu}_{\xmisi}\big)^\top \boldsymbol{\Omega}_{\xmisi}^{-1} \big(\xmisi-\boldsymbol{\mu}_{\xmisi}\big) \Bigg\}
    \cdot I(\xmisi \leq \v c_i)
\end{equation*}
where
\begin{align*}
     \boldsymbol{\mu}_{\xmisi} & = \Bigg(\frac{1}{\sigma^2}\v \beta_\mathrm{m} \v \beta_\mathrm{m}^\top + \bar{\v \Sigma}_i^{-1}\Bigg)^{-1} \Bigg(\bar{\v \Sigma}_i^{-1} \bar{\v \mu}_i+\frac{\tilde{y}_i}{\sigma^2} \v \beta_\mathrm{m}\Bigg)\\
    \boldsymbol{\Omega}_{\xmisi} & = \Bigg(\frac{1}{\sigma^2}\v \beta_\mathrm{m} \v \beta_\mathrm{m}^\top + \bar{\v \Sigma}_i^{-1}\Bigg)^{-1}.
\end{align*}
The full conditional posterior is therefore the truncated multivariate normal density given in Section \ref{sec:inference}.

\subsection{Data generation and prior hyperparameters in the experiments}\label{Appendix:data}

Section \ref{sec:simulations} use data simulated from the following process. 

The elements of $\v \beta$ are simulated independently from a standard normal distribution, and we set the error variance $\sigma^2 = 4$. When auxiliary variables $\v w$ are included in the model, $\v w \sim \mathcal{N}(\mathbf{0}, \v I_q)$ and each element of $\v \Gamma$ is simulated from a standard normal distribution. The covariance matrix $\v \Omega$ has a block-equicorrelation structure in the following way. The covariates are divided into two uncorrelated groups with unit standard deviations, correlation $\rho_1$ between all covariates in the first group and $\rho_2$ between covariates in the second group. Exact numbers of the correlations are given in the paper as they differ across experiments. This simulation setup gives regressions with $R^2$ in the range $0.65< R^2 < 0.8$. For the simulations without auxiliary variables, the covariates are drawn from the $\mathcal{N}(\v 0,\v \Omega)$ distribution, with the same $\v \Omega$ as described above.

The priors in the experiments are generally non-informative with $\tau_{\beta_0} = \tau_{\tilde \beta} = 10^4$ in the prior for $\beta_0$ and $\v \beta$, respectively. The prior hyperparameters of $\sigma^2$ are $a = m^2/v+2$, $b = m(m^2/v+1)$, where $v = 2000^2$ and $m$ is the variance of $\v y$, which also varies across experiments. For $\v \Omega$, we use a prior with $\kappa = 10$ degrees of freedom and set $A = \frac{1}{10}  \v I_p$. Finally, $\tau_\gamma = \sqrt{0.1}$ in the prior of $\v \gamma$.  
\end{document}